\documentclass[journal]{IEEEtran}

\usepackage{amsmath,amstext,amsfonts,amssymb,eucal,graphicx}
\usepackage{subfigure}
\usepackage{epsfig}
\usepackage{listings}

\newtheorem{proposition}{Proposition}
\newtheorem{theorem}{Theorem}
\newtheorem{definition}{Definition}

\newcommand{\E}{\mathbb{E}}

\begin{document}
\bibliographystyle{IEEEtran}
\title{Finite Dimensional Statistical Inference}
\author{
  \IEEEauthorblockN{\O yvind
  Ryan,~\IEEEmembership{Member,~IEEE}, Antonia Masucci, Sheng Yang,~\IEEEmembership{Member,~IEEE}, and M{\'e}rouane~Debbah,~\IEEEmembership{Senior Member,~IEEE}\\}
  \thanks{This work was supported by Alcatel-Lucent within the Alcatel-Lucent Chair on flexible radio at SUPELEC}
  \thanks{This paper was presented in part at the International Conference on Ultra Modern Telecommunications, 2009, St. Petersburg, Russia}
  \thanks{Antonia Masucci is with SUPELEC, Gif-sur-Yvette, France, antonia.masucci@supelec.fr}
  \thanks{\O yvind~Ryan is with the Centre of Mathematics for Applications, University of Oslo, P.O. Box 1053 Blindern, NO-0316 Oslo, Norway, and with SUPELEC, Gif-sur-Yvette, France, oyvindry@ifi.uio.no}
  \thanks{Sheng Yang is with SUPELEC, Gif-sur-Yvette, France, sheng.yang@supelec.fr}
  \thanks{M{\'e}rouane~Debbah is with SUPELEC, Gif-sur-Yvette, France, merouane.debbah@supelec.fr}
}

\markboth{IEEE Transactions on Information Theory,submitted,~December~2009}{Shell \MakeLowercase{\textit{et al.}}: Bare
Demo of IEEEtran.cls for Journals}

\maketitle
\begin{abstract}
In this paper, we derive the explicit series expansion of the
eigenvalue distribution of various models, namely the case of
non-central Wishart distributions, as well as correlated zero mean Wishart distributions.
The tools used extend those of the free probability framework, which have been quite successful for high dimensional statistical inference
(when the size of the matrices tends to infinity), also known
as free deconvolution. This contribution focuses on the finite
Gaussian case and proposes algorithmic methods to compute the
moments. Cases where asymptotic results fail to apply are also discussed.
\end{abstract}


\section{Introduction}
Random matrix and free probability theory have fruitful applications in many fields of research,
such as digital communication~\cite{paper:telatar99}, mathematical finance~\cite{book:bouchaud} and nuclear physics~\cite{paper:guhr}.
In particular, the free probability framework~\cite{vo2,paper:vomult,vo6,vo7,book:hiaipetz} can be used
for high dimensional statistical inference (or free deconvolution), i.e., to retrieve the eigenvalue distributions of involved functionals of random matrices.
The general idea of deconvolution is related to the following problem~\cite{Florent}:

Given ${\bf A}$, ${\bf B}$ two $n\times n$ independent square Hermitian (or symmetric)
random matrices:\\
1) Can one derive the  eigenvalue distribution of ${\bf A}$ from the ones of ${\bf A} + {\bf B}$ and ${\bf B}$?
If feasible in the large $n$-limit, this operation is named additive free deconvolution,\\
2) Can one derive the eigenvalue distribution of ${\bf A}$ from the ones
of ${\bf AB}$ and ${\bf B}$? If feasible in the large $n$-limit, this operation is named multiplicative free deconvolution.

In the literature, deconvolution  for the large $n$-limit
has been studied, and the methods generally used to compute it are the method of moments~\cite{vo2}, and the Stieltjes transform method~\cite{paper:doziersilverstein1}.
The expressions turn out to be quite simple if some kind of asymptotic freeness~\cite{book:hiaipetz} of the matrices involved is assumed.
However, freeness usually does not hold for finite matrices.
Quite remarkably, the method of moments can still be used to propose an algorithmic method
to compute these operations. The goal of this contribution is exactly to propose a general finite dimensional statistical inference framework
based on the method of moments, which is implemented in software.
As the calculations are quite tedious, and for sake of clarity, we focus in this contribution on
Gaussian matrices\footnote{Cases such as Vandermonde matrices can also
be implemented in the same vein~\cite{ryandebbah:vandermonde1,ryandebbah:vandermonde2}. The
general case is, however, more difficult.}.

The method of moments~\cite{Florent} is based on the relations between the moments of the different matrices involved.
It provides a series expansion of the eigenvalue distribution of the involved matrices.
For a given $n\times n$ random matrix ${\bf A}$, the $p$-th moment of ${\bf A}$ is defined as
\begin{equation} \label{moment}
  t_{{\bf A}}^{n,p} = \E\left[ \mathrm{tr}({\bf A}^p) \right]=\int \lambda^pd\rho_n(\lambda)
\end{equation}
where $\E$ is the expectation, $\mathrm{tr}$ the normalized trace,
and $d\rho_n$ the associated empirical mean measure defined by $d\rho_n(\lambda)=\E\left(\frac{1}{n} \sum_{i=1}^n \delta(\lambda-\lambda_i)\right)$,
where $\lambda_i$ are the eigenvalues of ${\bf A}$.
Quite remarkably, when $n\to \infty$, $t_{{\bf A}}^{n,p}$ converges in many cases almost surely to an analytical expression $t_{{\bf A}}^{p}$
that depends only on some specific parameters of ${\bf A}$
(such as the distribution of its entries)\footnote{Note that in the following, when speaking of moments of matrices, we refer to the moments of the associated measure.}.
This enables to reduce the dimensionality of the problem and  simplifies the computation of convolution of measures.
In recent works deconvolution has been analyzed when $n\to \infty$ for some particular matrices ${\bf A}$ and ${\bf B}$,
such as when  ${\bf A}$ and ${\bf B}$ are free~\cite{eurecom:freedeconvinftheory},
or ${\bf A}$ random Vandermonde and ${\bf B}$ diagonal~\cite{ryandebbah:vandermonde1,ryandebbah:vandermonde2}.

The inference framework  described in this contribution is based on the method of moments in the finite case: 
it takes a set of moments as input, and produces a set of moments as output, with the dimensions of the matrices considered finite.
The framework is flexible enough to allow for repeated combinations of the random matrices we consider,
and the patterns in such combinations are reflected nicely in the algorithms.
The framework also lends itself naturally to combinations with other types of random matrices,
for which support has already been implemented in the framework~\cite{ryandebbah:vandermonde2}.
This flexibility, exploited with the method of moments, is somewhat in contrast to methods such as the Stieltjes transform method~\cite{paper:doziersilverstein1},
where combining patterns of matrices naturally leads to more complex equations for the Stieltjes transforms (when possible) 
and can only be performed in the large $n$-limit. 
While the simplest patterns we consider are sums and products,
we also consider products of many independent matrices.
The algorithms are based on iterations through partitions and permutations as in~\cite{paper:haagerupthorbjornsen1},
where the case of a Wishart matrix was considered.
Our methods build heavily on the simple form which the moments of complex Gaussian random variables have, as exploited in~\cite{paper:haagerupthorbjornsen1}.
We remark that, in certain cases, it is possible to implement the method of moments in a different way also~\cite{haagerup98random,paper:tucci1}.
However, we are not aware of any attempts to make an inference framework as general as the one presented here. 
The case presented in~\cite{paper:tucci1}, for instance, handles only certain zero-mean, one-sided correlated Wishart matrices.

The paper is organized as follows.
Section~\ref{section:essentials} provides background essentials on random matrix theory and combinatorics needed to state the main results.
Parts of Section~\ref{section:essentials} is rather technical, 
but it is not necessary to understand all details therein to understand the statement of the main results. 
These are summarized in Section~\ref{section:theorems}.
First, algorithms for the simplest patterns (sums and products of random matrices) in the finite dimensional statistical inference framework are presented.
Then, recursive algorithms for products of many Wishart matrices and a deterministic matrix are included, as well with some general remarks on how the general situation
can be attacked from these basic algorithms.
We then explain how algorithms for deconvolution can be obtained within the same framework, and formalize the corresponding moment estimators.
Section~\ref{software} presents details on the software implementation of the finite dimensional statistical inference framework.
Section~\ref{simulations} presents some simulations and useful applications showing the implications of the presented results in various applied fields.

\section{Random matrix Background Essentials} \label{section:essentials}
In the following, upper boldface symbols will be used for matrices,
whereas lower symbols will represent scalar values.
$(.)^T$ will denote the transpose operator, $(.)^\star$
conjugation, and $(.)^H=\left((.)^T\right)^\star$ hermitian
transpose. ${\bf I}_n$ will represent the $n\times n$ identity matrix.
We let $\mathrm{Tr}$ be the (non-normalized) trace for square matrices, defined by,
\[
  \mathrm{Tr}({\bf A}) = \sum_{i=1}^n a_{ii},
\]
where $a_{ii}$ are the diagonal elements of the $n\times n$ matrix ${\bf A}$.
We also let $\mathrm{tr}$ be the normalized trace, defined by $\mathrm{tr}({\bf A}) = \frac{1}{n}\mathrm{Tr}({\bf A})$.
When ${\bf A}$ is non-random, there is of course no need to take the expectation in (\ref{moment}).
${\bf D}$ will in general be used to denote such non-random matrices,
and if ${\bf D}_1,\ldots,{\bf D}_r$ are such matrices, we will write
\begin{equation} \label{alphadef}
  D_{i_1,\ldots,i_s} = \mathrm{tr}\left( {\bf D}_{i_1}\cdots {\bf D}_{i_s}\right),
\end{equation}
whenever $1\leq i_1,\ldots,i_s\leq r$. (\ref{alphadef}) are also called {\em mixed moments}.

To state the results of this paper, random matrix concepts will be combined
with concepts from partition theory. ${\cal P}(n)$ will denote the partitions of
$\{1,\ldots,n\}$. For a partition $\rho=\{ W_1,\ldots,W_r\}\in{\cal P}(n)$, $W_1,\ldots,W_r$ denote its blocks,
while $|\rho|=r$ denotes the number of blocks.
We will write $k\sim_{\rho}l$ when $k$ and $l$ belong to the same block of $\rho$.
Partition notation is adapted to mixed moments in the following way:
\begin{definition} \label{ddef2}
  For $\rho = \{ W_1,\ldots,W_k \}$, with $W_i = \{ w_{i1},\ldots,w_{i|W_i|} \}$,
  we define
  \begin{align}
    D_{W_i}  &= D_{i_{w_{i1}},\ldots,i_{w_{i|W_i|}}} \label{dblockdef} \\
    D_{\rho} &= \prod_{i=1}^k D_{W_i}. \label{dpartdef}
  \end{align}
\end{definition}

With the empirical eigenvalue distribution of a hermitian random matrix ${\bf A}$, we mean the (random) function
\begin{equation} \label{edfdef}
  F_{ {\bf A} }(\lambda) = \frac{\#\{ i | \lambda_i \leq \lambda \}}{n} ,
\end{equation}
where $\lambda_i$ are the (random) eigenvalues of ${\bf A}$.
In many cases, the moments determine the distribution of the eigenvalues~\cite{book:baisilverstein}.
Due to the expectation in (\ref{moment}), the results in this paper thus apply to the mean eigenvalue distribution of certain random matrices.

In the following, we will denote a standard complex Gaussian matrix by ${\bf X}$. 
Standard complex means that the matrix
has i.i.d. complex Gaussian entries, with zero mean and unit variance (in particular, the real and imaginary parts of the entries are independent,
each with mean $0$, variance $\frac{1}{2}$).
${\bf X}$ will sometimes also be used to denote a standard selfadjoint Gaussian matrix, standard selfadjoint meaning that
it has i.i.d. entries only above or on the main diagonal,
with the real and imaginary parts independent with variance $\frac{1}{2}$~\cite{book:hiaipetz}.
The matrix sizes in the following will be denoted $n\times N$ for rectangular matrices, $n\times n$ for square matrices.
All random matrices we consider will be using selfadjoint or complex Gaussian matrices as building blocks.

\subsection{The diagrammatic method}
Some schools of science learn methods for computing the moments of Gaussian matrices from diagrams. As an example of what we mean by this, we have in
Figure~\ref{fig:momwishart} demonstrated how the second moment of a Wishart matrix 
$\frac{1}{N}{\bf X}{\bf X}^H$ can be found in this way. 
\begin{figure}
\begin{center}
  \subfigure[$2,3$ identified, $4,1$ also]{\epsfig{figure=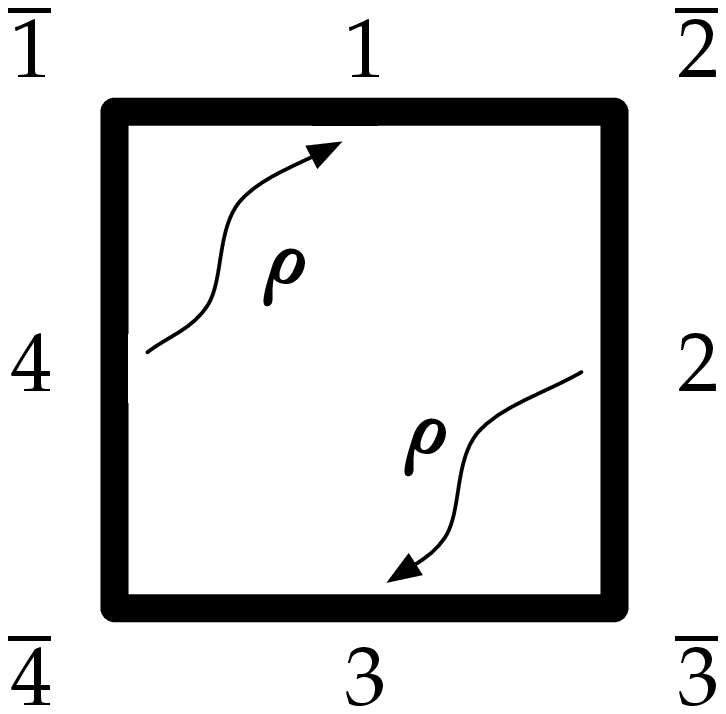,width=0.24\columnwidth}}
  \subfigure[$2,1$ identified, $4,3$ also]{\epsfig{figure=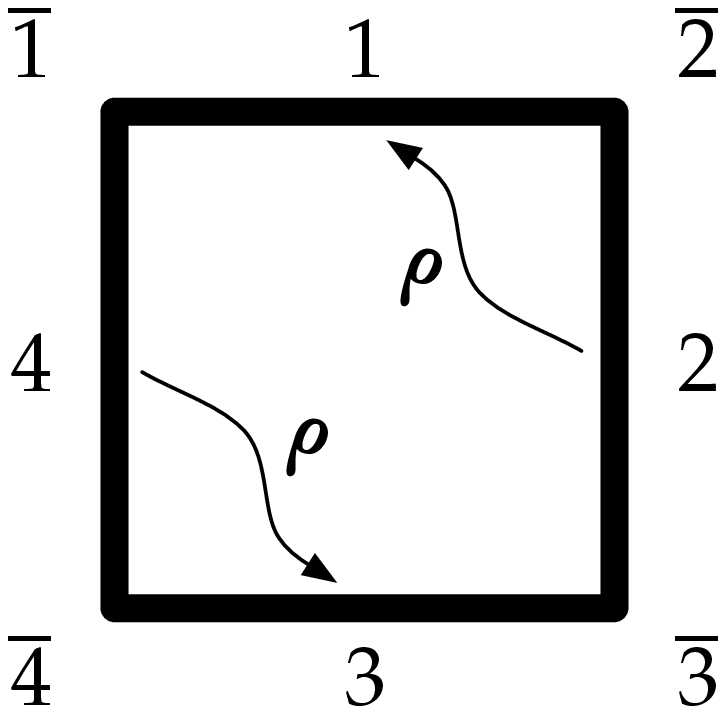,width=0.24\columnwidth}}
  
  \subfigure[Graph resulting from (a).]{\epsfig{figure=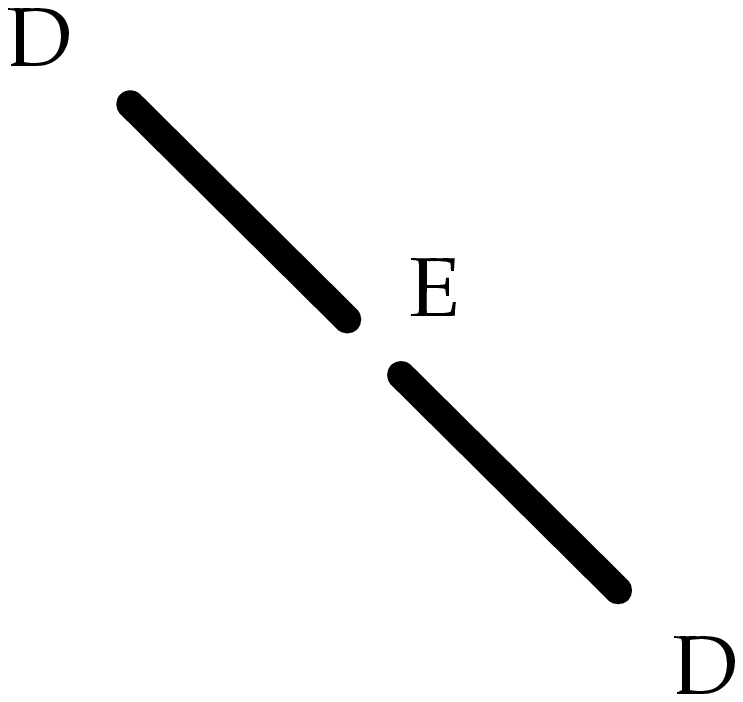,width=0.24\columnwidth}}
  \subfigure[Graph resulting from (b).]{\epsfig{figure=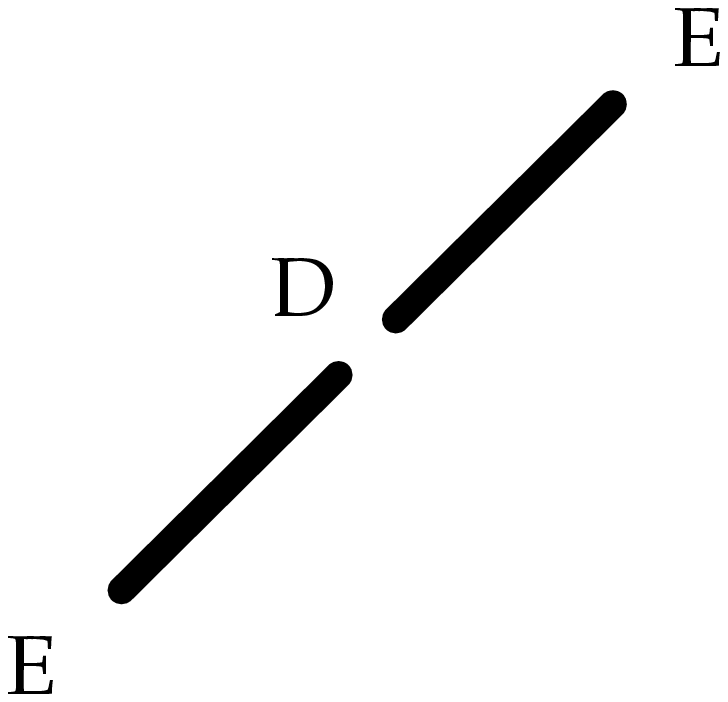,width=0.24\columnwidth}}
\caption{Diagrams demonstrating how the second moment of a Wishart matrix can be computed. 
         In all possible ways, even-labeled edges are identified with odd-labeled edges. 
         There are two possibilities, shown in (a) and (b). 
         The resulting graphs after identifications are shown in (c) and (d). 
         The second moment is constructed by summing contributions from all such possible identifications (here there are only $2$). 
         The contribution for any identification depends only on $n,N$, and the number of even-labeled and odd-labeled vertices in the resulting graphs (c) and (d).
         We will write down these contributions later.
         The labels $D$ and $E$ are included since we later on will generalize to compute the moments of doubly correlated Wishart matrices, 
         where the correlation matrices are denoted ${\bf D}$ and ${\bf E}$. 
         } \label{fig:momwishart}
\end{center}     
\end{figure}
In Figure~\ref{fig:momsamplcov} we have similarly demonstrated how the second moment of a matrix on the 
form $({\bf D}+{\bf X})({\bf E}+{\bf X})^H$ can be found, where ${\bf D}$ and ${\bf E}$ are independent from ${\bf X}$. 
\begin{figure}
  \subfigure[$4,1$ identified]{\epsfig{figure=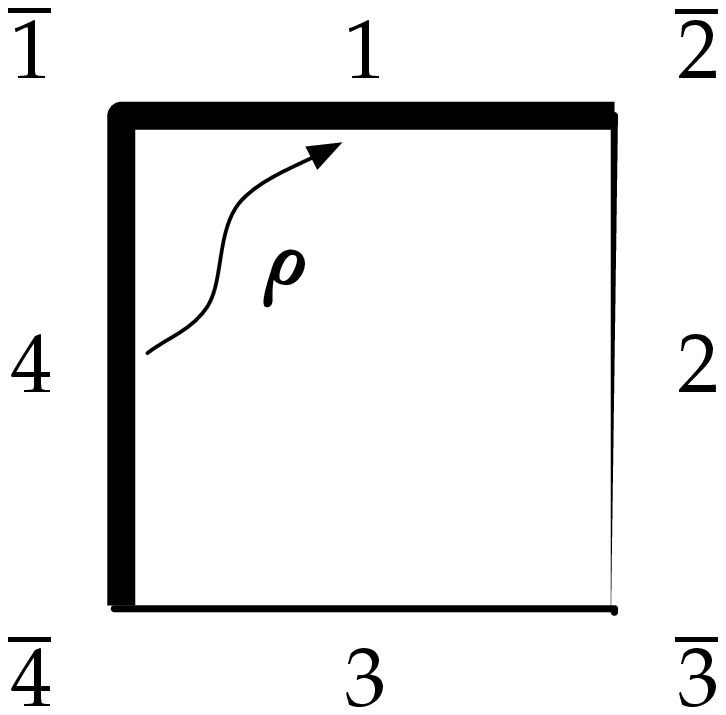,width=0.24\columnwidth}}
  \subfigure[$4,3$ identified]{\epsfig{figure=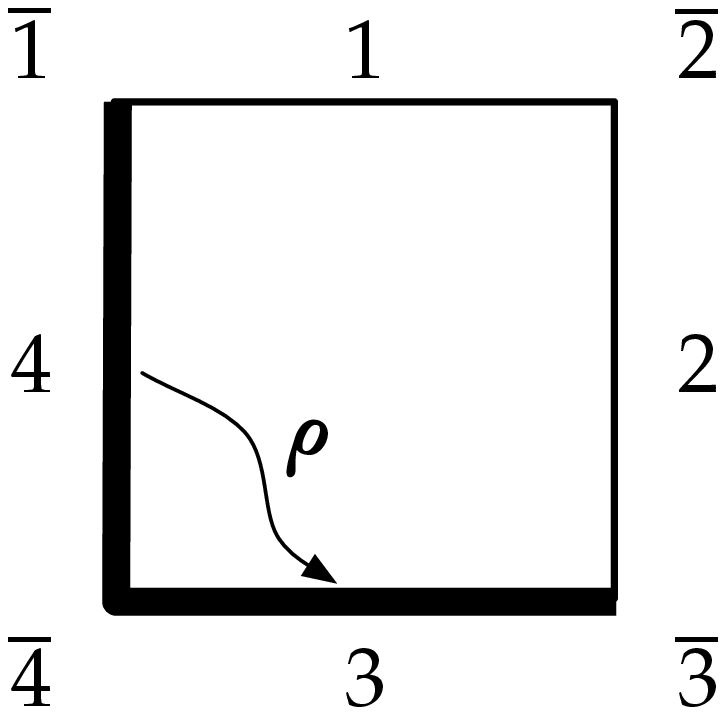,width=0.24\columnwidth}}
  \subfigure[$2,1$ identified]{\epsfig{figure=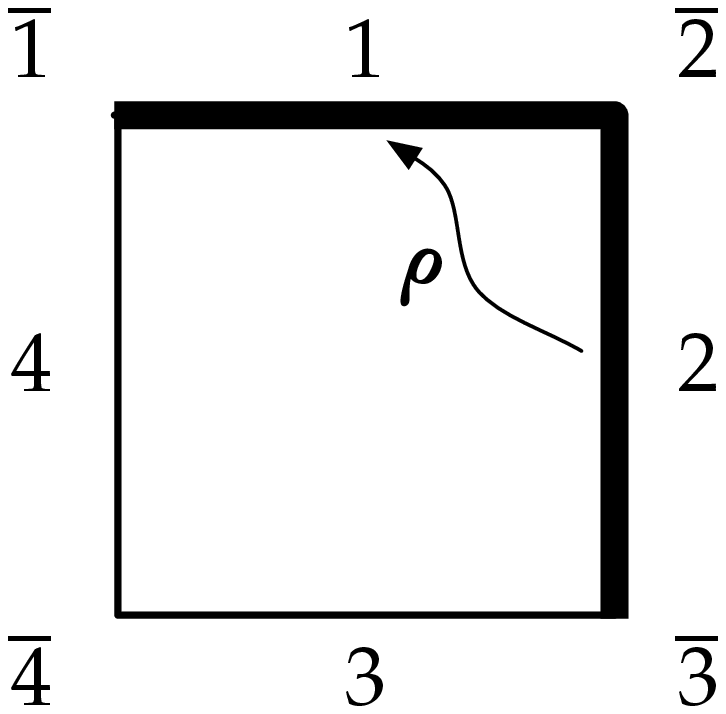,width=0.24\columnwidth}}
  \subfigure[$2,3$ identified]{\epsfig{figure=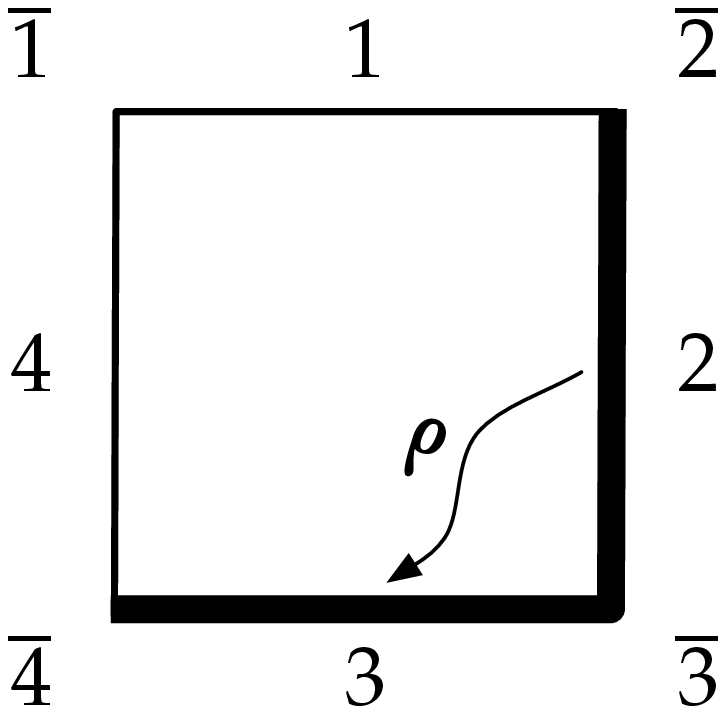,width=0.24\columnwidth}}
  
  \subfigure[Graph resulting from (a).]{\epsfig{figure=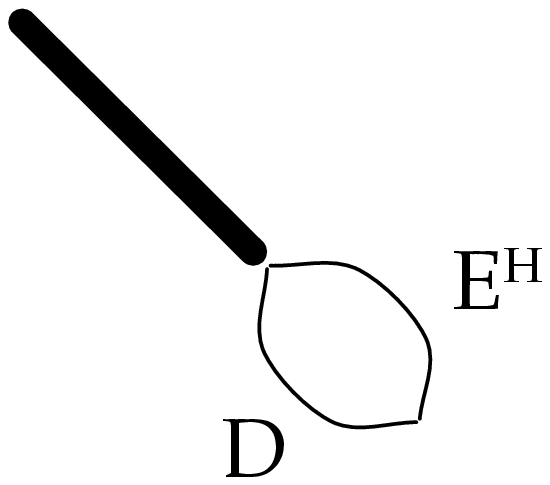,width=0.24\columnwidth}}
  \subfigure[Graph resulting from (b).]{\epsfig{figure=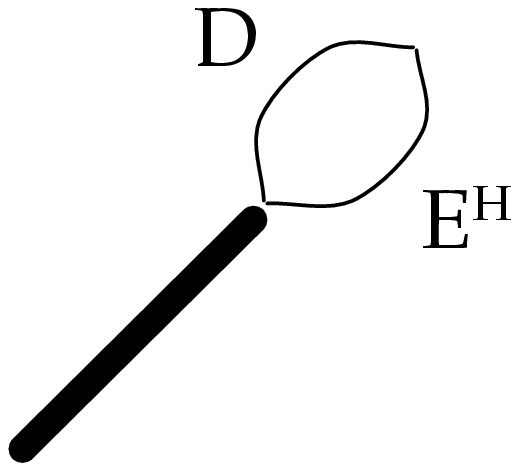,width=0.24\columnwidth}}
  \subfigure[Graph resulting from (c).]{\epsfig{figure=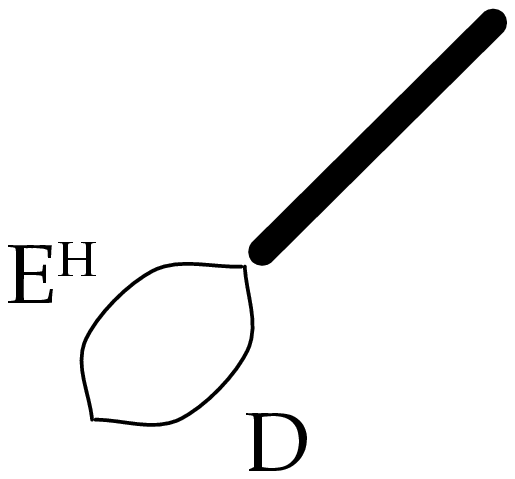,width=0.24\columnwidth}}
  \subfigure[Graph resulting from (d).]{\epsfig{figure=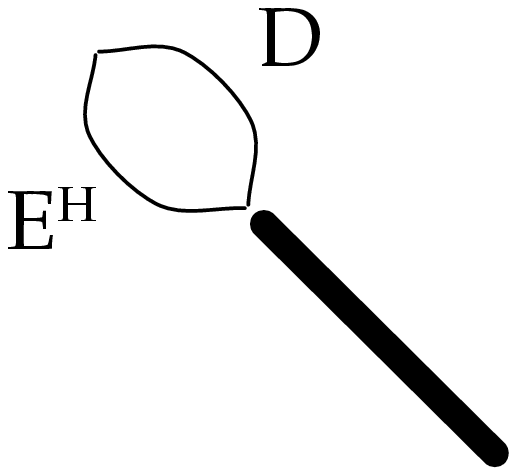,width=0.24\columnwidth}}
\caption{Diagrams demonstrating how the second moment of $({\bf D}+{\bf X})({\bf E}+{\bf X})^H$ can be computed. 
         As in Figure~\ref{fig:momwishart}, even-labeled and odd-labeled edges are identified in all possible ways, 
         but this time we also perform identifications of subsets of edges (edges not being identified correspond to choices from ${\bf D}$ and ${\bf E}^H$).
         In addition to the identifications in Figure~\ref{fig:momwishart}, we thus also have the ones in (a)-(d), 
         where only half of the edges are identified. We also have the case where there are no identifications at all. 
         The second moment is constructed by summing contributions from all such possible "partial" identifications, and
         the contribution for any identification is computed similarly as with Figure~\ref{fig:momwishart}.
         } \label{fig:momsamplcov}
\end{figure}
This matrix form is much used when combining many observations of a random vector.
While these two figures assume a complex Gaussian matrix, Figure~\ref{momselfadj} explains how the diagrammatic method can be modified to compute
the second moment of ${\bf R}+{\bf X}$, where ${\bf X}$ is selfadjoint, Gaussian, and ${\bf R}$ is independent from it. 
\begin{figure}
\begin{center}
  \subfigure[$2,4$ identified]{\epsfig{figure=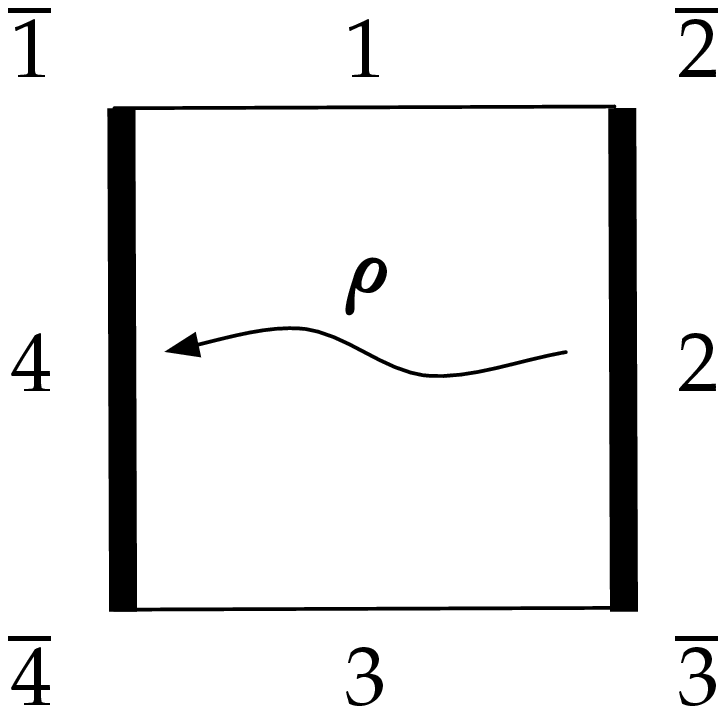,width=0.24\columnwidth}}
  \subfigure[$2,4$ identified, $1,3$ also]{\epsfig{figure=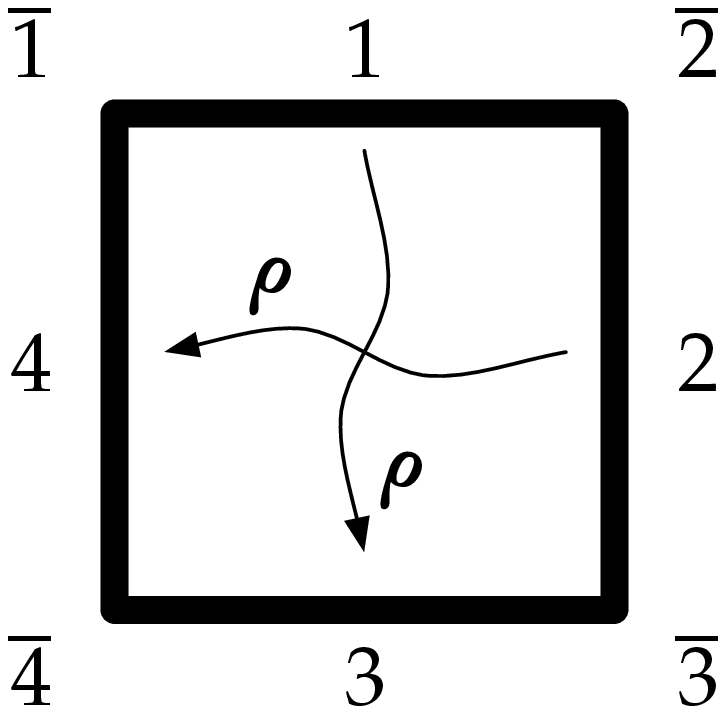,width=0.24\columnwidth}}
  
  \subfigure[Graph resulting from (a).]{\epsfig{figure=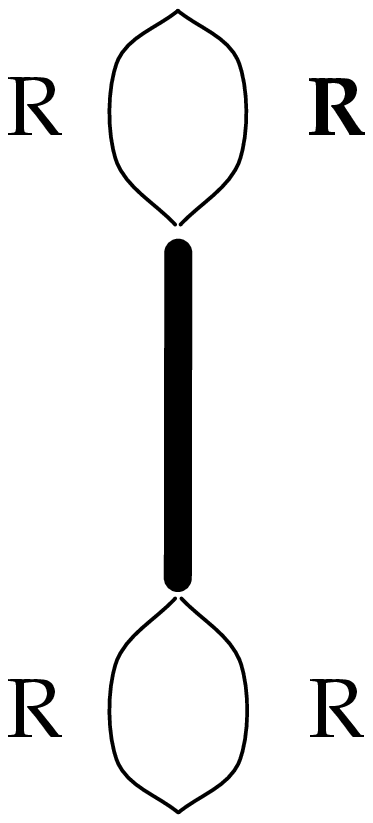,width=0.15\columnwidth}}
  \subfigure[Graph resulting from (b).]{\epsfig{figure=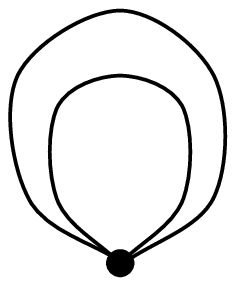,width=0.24\columnwidth}}
  \caption{Diagrams demonstrating how the second order moment of ${\bf R}+{\bf X}$ can be found, with ${\bf X}$ a selfadjoint, Gaussian matrix.
           As in Figures~\ref{fig:momwishart} and~\ref{fig:momsamplcov}, all possible identifications of edges are considered, but irrespective of whether they are even-odd pairings. 
           In particular, all identifications from these figures are considered, with the arrows allowed to go any way. 
           In addition, we also get identifications like in (a) and (b), where odd-labeled edges are identified, or even-labeled edges are identified.} \label{momselfadj}
\end{center}
\end{figure}
The diagrammatic method is easily generalized to higher moments, and to other random matrix models where Gaussian matrices are building blocks. 

The simple ingredient behind the diagrammatic method is the fact that one only needs consider conjugate pairings of complex Gaussian
elements~\cite{paper:haagerupthorbjornsen1}, which simplifies the computation of moments to simple identification of edges in graphs in all possible ways, as illustrated. 
This simple fact will be formalized in the following combinatorial definitions, which will be needed for the main results. 
The stated formulas are not new, since it has been known for quite some time that the diagrammatic method
can be used to obtain them. The value in this paper therefore does not lie in these formulas, 
but rather in making the general results possible to write down within a
framework, and available for computation in terms of an accompanying software implementation. 

Without going in all the details, there are similarities with the sketched diagrammatic approach, and other approaches based on diagrammatics. 
In particular in physics, and especially the field of statistical mechanics (see e.g.~\cite{paper:mezard,paper:nishimori}). 
It has been used recently in the field of wireless communications, related to the analysis of the mean and the variance 
of the Signal to Noise Ratio at the output of the MMSE receiver in MIMO and OFDM-CDMA systems~\cite{paper:moustakasdebbah}. 
Instead of calculating all the moments individually, one can represent these operations diagrammatically by solid lines and dashed lines. 
The idea is to draw them using {\em Feynman rules} derived from a generating function, and perform a resummation of all relevant graphs 
where averaging over matrices corresponds to connecting in all possible ways the different lines seperately. 
In many cases, in the large $N$-limit, only terms with non-crossing lines survive, A general description is proposed in~\cite{paper:brezin,paper:argaman,paper:brouwer}. 
The nomenclature we use for stating our results deviate some from that found in the literature.

To explain better how the diagrammatic method is connected to random matrices, 
write the trace of a product of matrices as
\begin{eqnarray}
\lefteqn{\E\left[ \mathrm{tr}({\bf A}_1{\bf A}_2\cdots {\bf A}_p) \right]} \nonumber \\
&=& \frac{1}{n}\sum_{i_1,i_2,...,} a^{(1)}(i_1,i_2)a^{(2)}(i_2,i_3)\cdots a^{(p)}(i_p,i_1), \label{tracewrittenout}
\end{eqnarray}
where the entries of ${\bf A}_k$ are $a^{(k)}(i,j)$.
We will visualize the matrix indices $i_1,...i_p$ as points (in the following also called vertices)
$\bar{1},...,\bar{p}$ on a circle, and the matrix entries
$a^{(1)}(i_1,i_2),...,a^{(p)}(i_p,i_1)$ as edges labeled $1,...,p$, with the
points $\bar{k},\overline{k+1}$ being the end points of the edge labeled $k$. 
We will call this the circular representation of (\ref{tracewrittenout}).
If ${\bf X}$ is $n\times N$ standard, complex, Gaussian, 
the circular representation of $\E\left(\text{tr}(({\bf X}{\bf X}^{H})^4)\right)$, 
before any Gaussian pairings have taken place, is thus shown in
Figure~\ref{fig:geo_interp1}. 
\begin{figure}[!h]
      \centering\epsfig{figure=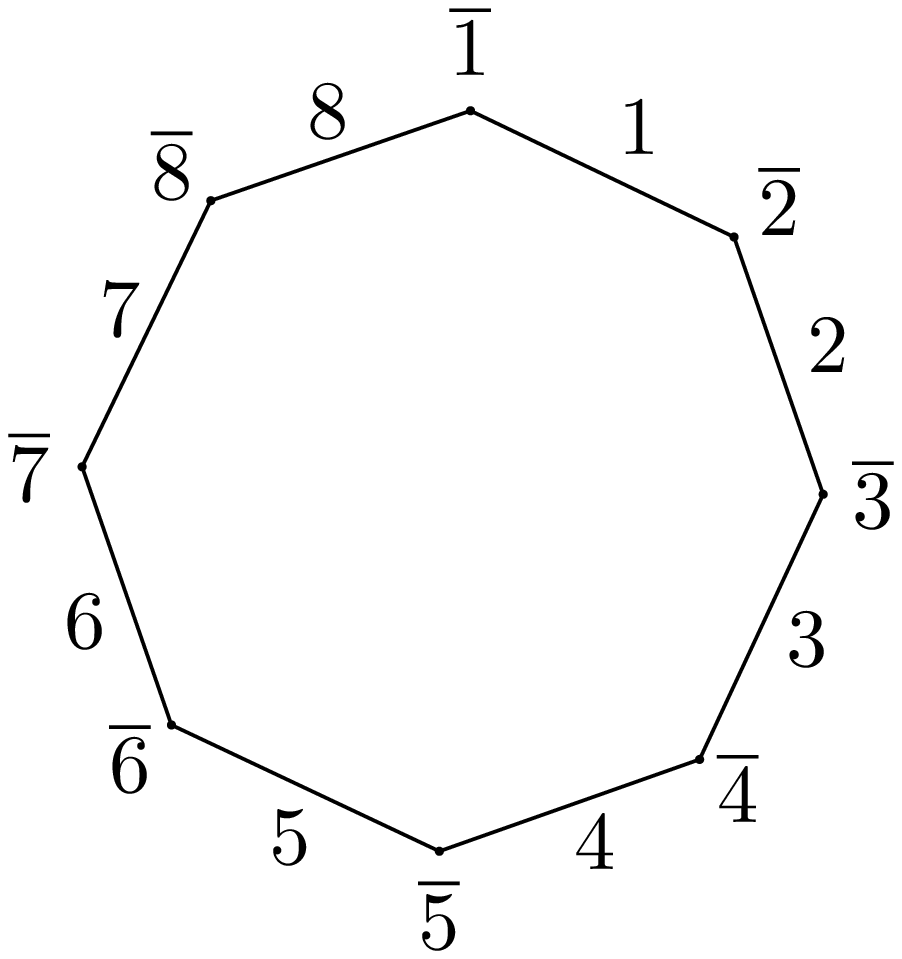,width=0.4\columnwidth}
        \caption{Ordering of points and edges on a circle for the trace $\E\left(\text{tr}(({\bf X}{\bf X}^{H})^4)\right)$ 
        for ${\bf X}$ complex, standard, Gaussian, when it is written out as in
        (\ref{tracewrittenout}). The odd edges $1,3,5,7$ correspond to terms of the form ${\bf X}$; the even edges
        $2,4,6,8$ correspond to terms of the form ${\bf X}^H$; the odd vertices
        $\overline{1},\overline{3},\overline{5},\overline{7}$ correspond
        to a choice among $1,\ldots,n$ (i.e. among row indices in ${\bf X}$) ;
        the even vertices $\overline{2},\overline{4},\overline{6},\overline{8}$ correspond
        to a choice among $1,\ldots,N$ (i.e. among column indices in
        ${\bf X}$). The bars, used to differ between edges and vertices, are
        only used in the figures.}
 \label{fig:geo_interp1}
\end{figure}

More general than (\ref{tracewrittenout}), we can have a product of $k$ traces,
\begin{equation} \label{moregeneralmoment}
\E\left[ 
\mathrm{tr}\left({\bf A}_1                     \cdots {\bf A}_{p_1}\right) \cdots 
\mathrm{tr}\left({\bf A}_{p_1+\cdots+p_{k-1}+1}\cdots {\bf A}_{p_1+\cdots+p_k}\right)\right].
\end{equation}
This will also be given an interpretation in terms of $k$ circles, 
with $p_1,...,p_k$ points/edges on each, respectively.
Conjugate pairings of complex Gaussian elements in (\ref{moregeneralmoment}) in all possible ways are performed as in the case of one circle, 
and is illustrated in Figure~\ref{fig:geo_interp4} for $k=2$ and $p_1=p_2=3$, 
with the first edge on the first cirle paired with the last edge on the second circle. 
In (\ref{moregeneralmoment}), this corresponds to $a^{(1)}(i_1,i_2)$ and
$a^{(12)}(i_{12},i_7)$ being conjugate of each other 
(there are twelve matrices present here,since ${\bf A}_i={\bf X}_i{\bf X}_i^H$).
This can only be the case if $i_1=i_7$ and $i_2=i_{12}$.
\begin{figure}[!h] \centering      \epsfig{figure=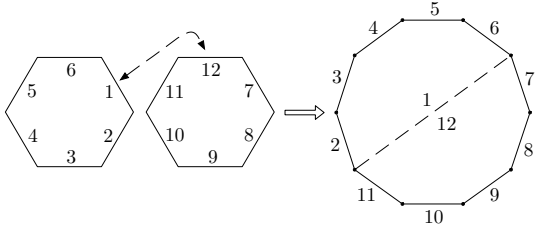,width=0.8\columnwidth}
  \caption{Identification of edges across two circles. Such identifications arise in the computations of (\ref{moregeneralmoment}). }
 \label{fig:geo_interp4}
\end{figure}

\subsection{Formalizing the diagrammatic method}
The following definition, which is a generalization from~\cite{paper:haagerupthorbjornsen1}, 
formalizes identifications of edges, as we have illustrated:
\begin{definition} 
  Let $p$ be a positive integer. By a partial permutation we mean a one-to-one mapping $\pi$
  between two subsets $\rho_1,\rho_2$ of $\{ 1,\ldots,p\}$.
  We denote by $\text{SP}_p$ the set of partial permutations of $p$ elements.
  When $\pi\in\text{SP}_p$, we define $\hat{\pi}\in \text{SP}_{2p}$ by
\begin{align*}
  \hat{\pi}(2j-1)&=2{\pi}^{-1}(j),\quad j\in\rho_2\\
  \hat{\pi}(2j)&=2\pi(j)-1,\quad j\in\rho_1.
\end{align*}
\end{definition}

Note that in this definition, subtraction is performed in such a way that the result stays within the same circle. 
In terms of Figure~\ref{fig:geo_interp4}, this means that $1-1=6,6-1=5$ and $7-1=12,12-1=11$. 
Addition is assumed to be performed in the same way, so that $1+1=2,6+1=1$, and $7+1=8,12+1=7$. 
In the following, this convention for addition and subtraction will be used, and the number of edges on the circles will be implicitly assumed,
and only mentioned when strictly needed.

When we compute $\mathrm{tr}((({\bf D}+{\bf X})({\bf E}+{\bf X})^H)^p)$, we multiply out to obtain a sum of 
terms of length $2p$ on the form (\ref{tracewrittenout}), where the terms are one of ${\bf X},{\bf X}^H,{\bf D}$, or ${\bf E}^H$, 
with $\cdot$- and $\cdot^H$-terms appearing in alternating order.
$\rho_1$ corresponds to the indices of ${\bf X}^H$ in such a term (after their order of appearance), 
$\rho_2$ to the indices of ${\bf X}$, and $\pi$ to the Gaussian conjugate pairings.
Computing $\mathrm{tr}((({\bf D}+{\bf X})({\bf E}+{\bf X})^H)^p)$ thus boils down to iterating through $\text{SP}_p$. 
In Figure~\ref{fig:momsamplcov}, this was examplified with $p=2$, with the sizes of the subsets equal to $1$. 
$\rho_1$ was indicated by the starting edges of the arrows, $\rho_2$ by the ending edges. 
(a) to (d) represents the only possible pairings in the terms 
${\bf X}{\bf E}^H{\bf D}{\bf X}^H$, ${\bf D}{\bf E}^H{\bf X}{\bf X}^H$, ${\bf X}{\bf X}^H{\bf D}{\bf E}^H$, and ${\bf D}{\bf X}^H{\bf X}{\bf E}^H$, respectively.
The case for a Wishart matrix is simpler, since there is no need to multiply out terms, and we need only consider $\pi\in\text{SP}_p$ where $|\rho_1|=|\rho_2|=p$, 
as shown in Figure~\ref{fig:momwishart}. Such $\pi$ are in one-to-one correspondence with ${\cal S}_p$, the set of permutations of $p$ elements.

It is clear that $\hat\pi$ maps $(2\rho_1)$ onto $\cup(2\rho_2-1)$, and has period two (i.e. $\hat{\pi}(\hat{\pi}(\rho))=\rho$ for all $\rho$),
where $2\rho_2-1=\{ 2k-1|k\in\rho_2\}$. 
In particular, $\hat\pi$ maps even numbers to odd numbers, and vice versa.
When edges are identified as dictated by a partial permutation, 
vertices are identified as dictated by the partition $\rho(\pi)$ defined as follows:
\begin{definition}\label{thisdef}
Let $\pi$ be a partial permutation,
and let $\hat\pi$ be determined by $\rho_1,\rho_2$ and a pairing between them.
We associate to $\pi$ an equivalence relation $\rho=\rho(\pi)$ on $\{ 1,...,2p\}$ generated by
\begin{equation}\label{equivalence2}
  j\sim_{\rho} \hat{\pi}(j) +1 \mbox{, } j+1\sim_{\rho} \hat{\pi}(j) \mbox{, for } j\in\rho_1.
\end{equation}
We let $k(\rho)$ and $l(\rho)$ denote the number of blocks of $\rho$ consisting of only even or odd numbers, respectively.
\end{definition}

Any block in $\rho$ consists either of even numbers, or odd numbers, since $\hat\pi$ maps between even and odd numbers, 
so that the definitions of $k(\rho)$ and $l(\rho)$ above make sense.
In the following, we will let $k_1,\ldots,k_{k(\rho)}$ be the cardinalities of the blocks consisting of even numbers only,
and $l_1,\ldots,l_{l(\rho)}$ the cardinalities of the blocks consisting of odd numbers only. 
The restriction of $\rho$ to the odd numbers thus defines another partition, which we will denote $\rho|\text{odd}$.
Similarly, the restriction of $\rho$ to the even numbers yields another partition, which we will denote $\rho|\text{even}$.
$\rho|\text{odd}$ and $\rho|\text{even}$ will appear in the main results later on.

$\rho$ should be interpreted as an equivalence relation on matrix indices occuring in ${\bf X},{\bf X}^H$. 
The following definition similarly keeps track of how conjugate pairings group matrix indices occuring in ${\bf D},{\bf E}^H$ into traces:
\begin{definition} \label{sigmadef}
Let ${\cal D}\subset\{ 1,..,2p\}$ be the set of deterministic edges (i.e. edges corresponding to ocurrences of ${\bf D},{\bf E}^H$), 
and let $\pi\in\text{SP}_p$ be determined by $\rho_1,\rho_2$.
$\sigma=\sigma(\pi)$ is defined as the equivalence relation on ${\cal D}$ generated by the relations
\begin{eqnarray}
  k \sim_{\sigma} k+1 &\text{if}& k,k+1\in{\cal D} \label{eq1} \\
  k \sim_{\sigma} l   &\text{if}& k,l\in{\cal D},k+1\sim_{\rho}l. \label{eq2}
\end{eqnarray}
Let also $kd(\rho)$ be the number of blocks of $\rho$
contained within the even numbers which intersect ${\cal D}\cup({\cal D}+1)$,
and let $ld(\rho)$ be the number of blocks of $\rho$
contained within the odd numbers which intersect ${\cal D}\cup({\cal D}+1)$.
\end{definition}

Two edges from ${\cal D}$ belong to the same block of $\sigma$ if, after identifying edges, they are connected with a path of edges from ${\cal D}$. 
A block of $\rho$ which contains a vertex from ${\cal D}\cup({\cal D}+1)$ corresponds to a matrix index which occurs in a deterministic element. 
As an example, in Figure~\ref{fig:momsamplcov} all four partial permutations are seen to give rise to a $\sigma$ with one block only. 
They are seen to be $\sigma=\{2,3\}$ for (a), $\sigma=\{1,2\}$ for (b), $\sigma=\{3,4\}$ for (c), and $\sigma=\{1,4\}$ for (d).

\subsection{Formalizing the diagrammatic method for selfadjoint matrices}
A standard, selfadjoint, Gaussian $n\times n$ random matrix ${\bf X}$ can be written on the form ${\bf X}=\frac{1}{\sqrt{2}}({\bf Y}+{\bf Y}^H)$,
where ${\bf Y}$ is an $n\times n$ standard complex Gaussian matrix.
We can thus compute the moments of ${\bf R}{\bf X}$ and ${\bf R}+{\bf X}$ (with ${\bf X}$ selfadjoint Gaussian, and ${\bf R}$ selfadjoint and independent from ${\bf X}$) by
substituting this, and summing over all possible combinations of ${\bf Y}$ and ${\bf Y}^H$. This rewriting in terms of complex Gaussian matrices means that we need to
slightly change the definitions of the partitions $\rho$ and $\sigma$ to the following:
\begin{definition} \label{olddef}
Let $\pi\in \text{SP}_p$ be determined by
disjoint subsets $\rho_1,\rho_2$ of $\{1,\ldots,p\}$ with $|\rho_1|=|\rho_2|$ (in particular, $2|\rho_1|\leq p$).
We associate to $\pi$ an equivalence relation $\rho_{sa}=\rho_{sa}(\pi)$ on $\{ 1,...,2p\}$ generated by
\begin{eqnarray*}
  i \sim_{\rho_{sa}} \pi(i)+1      &\text{for}& i\in\rho_1, \\
  \pi^{-1}(i)+1 \sim_{\rho_{sa}} i &\text{for}& i\in\rho_2.
\end{eqnarray*}
\end{definition}

As before, $\rho_1$ corresponds to choices from ${\bf X}^H$, $\rho_2$ to choices from ${\bf X}$, 
when the selfadjoint Gaussian matrix is expressed as a sum of complex Gaussian matrices. 
Definition~\ref{sigmadef} is modified as follows:
\begin{definition} \label{newsigmadef}
With $\pi,\rho_1,\rho_2$ as in Definition~\ref{olddef},
$\sigma_{sa}=\sigma_{sa}(\pi)$ is defined as the equivalence relation on ${\cal D}=\left( \rho_1\cup\rho_2 \right)^c$ generated by the relations
\begin{eqnarray}
  k\sim_{\sigma_{sa}} k+1 &\text{if}&  k,k+1\in{\cal D} \label{eq1new} \\
  k\sim_{\sigma_{sa}} l   &\text{if}&  k,l\in{\cal D},k+1\sim_{\rho_{sa}}l \nonumber \\
                          &         &  \text{or }  k\sim_{\rho_{sa}}l+1. \label{eq2new}
\end{eqnarray}
Define also $d(\rho_{sa})$ as the number of blocks of $\sigma_{\rho_{sa}}$ which intersect ${\cal D}\cup({\cal D}+1)$.
\end{definition}

As an example, In Figure~\ref{momselfadj}(c) we have that $\rho_{sa}=\{\{1,2\},\{3,4\}\}$, $\sigma_{sa}=\{\{1\},\{3\}\}$, 
in Figure~\ref{momselfadj}(d) we have that $\rho_{sa}=\{\{1,2,3,4\}\}$, $\sigma_{sa}$ is the empty partition.

In the following, we will state our results in terms of normalized traces.
We remark that some of these have been stated previously in terms of non-normalized traces~\cite{paper:haagerupthorbjornsen1}.
In some results, we have substituted $c=\frac{n}{N}$, which makes the results compatible with the asymptotic case often used in the literature, 
where $n$ and $N$ grow to infinity at the same rate, the rate being $c=\lim_{n\to\infty}\frac{n}{N}$.
In the following, equivalence relations will interchangeably also be refered to as partitions.

\section{Statement of main results} \label{section:theorems}
The main results of the paper are split into three sections. 
In the first, basic sums and products are considered,
basic meaning that there is only one random matrix involved. 
In the second section we expand to the case where independent random matrices are involved, 
in which case expectations of products of traces are brought into the picture.
In these two sections, all Gaussian matrices are assumed complex and
rectangular, for which the results relate to the moments of the singular law of the matrices. 
In the third section we state similar results for the case where the Gaussian matrices instead are assumed square and selfadjoint.

\subsection{Basic sums and products}
Our first and simplest result concerns the moments of a doubly correlated Wishart matrix. 
These matrices are the most general known form we have found which have been considered in the literature~\cite{paper:bjw}, 
which can be addressed by our results: 
\begin{theorem}\label{teo1}
Let $n,N$ be positive integers, ${\bf X}$ be $n\times N$ standard, complex, Gaussian,
and ${\bf D}$ a (deterministic) $n\times n$ matrix,
${\bf E}$ a (deterministic) $N\times N$ matrix.
For any positive integer $p$,
\begin{eqnarray*}
\lefteqn{\E\left[\mathrm{tr}\left(\left(\frac{1}{N}{\bf D}{\bf X}{\bf E}{\bf X}^H\right)^p\right)\right]} \nonumber \\
&=& \sum_{\pi\in S_p} N^{k(\rho)-p} n^{l(\rho)-1} D_{\rho|\text{odd}} E_{\rho|\text{even}}. \label{momentsDXX^H}
\end{eqnarray*}
\end{theorem}

Theorem~\ref{teo1} is proved in Appendix~\ref{proofteo1}.
The next results will be proved using the same techniques, and will therefore be given shorter proofs.
The special case of a product of a Wishart matrix and a deterministic matrix, and a Wishart matrix itself, can be considered as a special case. 
It is seen that, for the latter, the contribution for the identification of edges refered to in Figure~\ref{fig:momwishart} is equal to $N^{k(\rho)-p} n^{l(\rho)-1}$, 
which indeed depends only on $n,N$, and the number of even-labeled and odd-labeled vertices in the resulting graphs. 
As an example, the contributions from the two possible identifications of edges giving the second moment of a Wishart matrix is 
$N^{1-2}n^{2-1}=\frac{n}{N}=c$ (Figure~\ref{fig:momwishart}(a)) and
$N^{2-2}n^{1-1}=1$ (Figure~\ref{fig:momwishart}(b)). The second moment is thus $1+c$, 
which also can be infered from the more general formuals in Section~\ref{software}.

We remark also that other closed forms of (\ref{momentsDXX^H}) can be found in the literature.
When the Wishart matrices are one-sided correlated (i.e. ${\bf E}=I$),~\cite{paper:tucci1} gives us the means to find the first order moments 
(i.e. one circle only is involved) in certain cases, also if the $p$'th moment is replaced with more general functionals of ${\bf X}$. It seems, however, 
that this result and the techniques used to prove it are hard to generalize.

We now turn to the moments of $({\bf D}+{\bf X})({\bf E}+{\bf X})^H$.
In the large $n,N$-limit, the case where ${\bf D}={\bf E}$ is related to the concept of rectangular free convolution~\cite{benaychgeorges1},
which admits a nice implementation in terms of moments~\cite{eurecom:multfreeconv}.
When $n$ and $N$ are finite, the following will be proved in Appendix~\ref{proofteo2}.

\begin{theorem} \label{teo2}
Let ${\bf X}$ be an $n\times N$ standard, complex, Gaussian matrix,
${\bf D},{\bf E}$ deterministic $n\times N$ matrices,
and set $D_p=\mathrm{tr}\left(\left(\frac{1}{N}{\bf D}{\bf E}^H\right)^p\right)$.
We have that
\begin{eqnarray}\label{genformula}
  \lefteqn{\E\left[ \mathrm{tr}\left(\left( \frac{1}{N} ({\bf D}+{\bf X})({\bf E}+{\bf X})^H \right)^p\right) \right]}  \nonumber \\
  &=& \sum_{{\pi\in \text{SP}_p}\atop{\pi=\pi(\rho_1,\rho_2,q)}} \frac{1}{nN^{|\rho_1|}} N^{k(\rho(\pi))-kd(\rho(\pi))} \nonumber \\
  & & \hspace{1cm} \times n^{l(\rho(\pi))-ld(\rho(\pi))} \nonumber \\
  & & \hspace{1cm} \times n^{|\sigma(\pi)|} \prod_i D_{|\sigma(\pi)_i|/2}.
\end{eqnarray}
\end{theorem}

Note that in Theorem~\ref{teo2}, $n$- and $N$-terms have not been grouped together.
This has been done to make clear in the proof the origin of the different terms.

\subsection{Expectations of products of traces}
Theorems~\ref{teo1} and~\ref{teo2} can be recursively applied, once one replaces ${\bf D}$ and ${\bf E}$ with random matrices.
In this process, we will see that expectations of products of traces are also needed, not only the first order moments as in theorems~\ref{teo1} and~\ref{teo2}. 
The recursive version of Theorem~\ref{teo1} looks as follows.

\begin{theorem} \label{recursive}
Assume that the $n\times n$ random matrix ${\bf R}$ and the $N\times N$ random matrix ${\bf S}$ are both independent from
the $n\times N$ standard, complex, Gaussian matrix ${\bf X}$, and define
\begin{align*}
  R_{l_1,...,l_r,m_1,...,m_s} &= \E\left[   \mathrm{tr}\left( {\bf R}^{l_1} \right) \mathrm{tr}\left( {\bf R}^{l_2} \right) \cdots \mathrm{tr}\left( {\bf R}^{l_r} \right)\right. \\
                              &      \left. \times\mathrm{tr}\left( {\bf S}^{m_1} \right) \mathrm{tr}\left( {\bf S}^{m_2} \right) \cdots \mathrm{tr}\left( {\bf S}^{m_s} \right)\right] \\
  M_{p_1,\ldots,p_k} &= \E\left[            \mathrm{tr}\left(\left( \frac{1}{N}{\bf R}{\bf X}{\bf S}{\bf X}^H \right)^{p_1}\right) \right.\\
                     &\qquad  \times        \mathrm{tr}\left(\left( \frac{1}{N}{\bf R}{\bf X}{\bf S}{\bf X}^H \right)^{p_2}\right) \cdots \\
                     &\qquad  \left. \times \mathrm{tr}\left(\left( \frac{1}{N}{\bf R}{\bf X}{\bf S}{\bf X}^H \right)^{p_k}\right) \right].
\end{align*}
Set $p=p_1+\cdots+p_k$, and let as before 
$l_1,\ldots,l_r$ be the cardinalities of the blocks of odd numbers only of $\rho$,
$m_1,\ldots,m_s$ be the cardinalities of the blocks of even numbers only of $\rho$,
with $k(\rho)$, $l(\rho)$ the number of blocks consisting of even and odd numbers only, respectively.
We have that
\begin{eqnarray}
  \lefteqn{M_{p_1,\ldots,p_k}} \nonumber \\
  &=& \sum_{\pi\in S_p} N^{k(\rho(\pi))-p}n^{l(\rho(\pi))-k} R_{l_1,...,l_r,m_1,...,m_s}. \label{recursiveformula}
\end{eqnarray}
\end{theorem}

\begin{proof}
There are only two differences from Theorem~\ref{teo1}. 
First, $n^{-1}$ is replaced with $n^{-k}$, since we now are taking $k$ traces instead of $1$ (we modify with additional trace normalization factors). 
Second, we replace the trace of a deterministic matrix with the expectation of a random matrix.
It is clear that the only additional thing needed for the proof to be replicated is that the random matrices ${\bf X}$ and ${\bf R}$ are independent.
\end{proof}

In some cases, for instance when ${\bf S}=I$, Theorem~\ref{recursive} also allows for deconvolution.
By this we mean that we can write down unbiased estimators for $R_{p_1,...,p_k}$, 
(which is the simplified notation for the mixed moments for the case where ${\bf S}=I$)
from an observation ${\bf Y}$ of $\frac{1}{N}{\bf R}{\bf X}{\bf X}^H$. 
This is achieved by stating all possible equations in (\ref{recursiveformula}) (i.e. for all possible $p_1,\ldots,p_k$), 
and noting that these express a linear relationship between all $\{R_{p_1,...,p_k}\}_{p_1,...,p_k}$, and all $\{M_{p_1,\ldots,p_k}\}_{p_1,...,p_k}$, 
where there are as many equations are unknowns. 
The implementation presented in Section~\ref{software} thus performs deconvolution by constructing the matrix corresponding to (\ref{recursiveformula}), 
and applying the inverse of this to the aggregate vector of all mixed moments of the observation.
We remark that the inverse may not exist if $N<n$, as will also be seen from expressions for these matrices in Section~\ref{software}.

It is also clear that the theorem can be recursively applied to compute the moments of any product of independent Wishart matrices
\begin{equation} \label{multprod}
  {\bf D}\frac{1}{N_1}{\bf X}_1{\bf X}_1^H\frac{1}{N_2}{\bf X}_2{\bf X}_2^H\cdots \frac{1}{N_k}{\bf X}_k{\bf X}^H_k,
\end{equation}
where ${\bf D}$ is deterministic and ${\bf X}_i$ is an $n\times N_i$
standard complex Gaussian matrix.
The ${\bf R}$'s during these recursions will simply be
\begin{align*}
  {\bf R}_1 &= {\bf D}\frac{1}{N_1}{\bf X}_1{\bf X}_1^H\frac{1}{N_2}{\bf X}_2{\bf X}_2^H\cdots \frac{1}{N_{k-1}}{\bf X}_{k-1}{\bf X}^H_{k-1} \\
  {\bf R}_2 &= {\bf D}\frac{1}{N_1}{\bf X}_1{\bf X}_1^H\frac{1}{N_2}{\bf X}_2{\bf X}_2^H\cdots \frac{1}{N_{k-2}}{\bf X}_{k-2}{\bf X}^H_{k-2} \\
  \vdots    &  \quad\ \vdots \\
  {\bf R}_k &= {\bf D}.
\end{align*}
Unbiased estimators for the moments of ${\bf D}$ from observations of the form (\ref{multprod}) can also be written down. 
Such deconvolution is a multistage process, where each stage corresponds to multiplication with an inverse matrix, 
as in the case where only one Wishart matrix is involved.

The recursive version of Theorem~\ref{teo2} looks as follows.

\begin{theorem} \label{recursivesum}
Let ${\bf X}$ be an $n\times N$ standard, complex, Gaussian matrix and let ${\bf R},{\bf S}$ be $n\times N$ and independent from ${\bf X}$. Set
\begin{align*}
  R_{p_1,\ldots,p_k} &= \E\left[        \mathrm{tr}\left( \left( \frac{1}{N}{\bf R}{\bf S}^H\right)^{p_1} \right) \right.\\
                  & \qquad \times       \mathrm{tr}\left( \left( \frac{1}{N}{\bf R}{\bf S}^H\right)^{p_2} \right) \cdots \\
                  & \qquad \left.\times \mathrm{tr}\left( \left( \frac{1}{N}{\bf R}{\bf S}^H\right)^{p_k} \right) \right] \\
  M_{p_1,\ldots,p_k} &= \E\left[           \mathrm{tr}\left( \left( \frac{1}{N}({\bf R}+{\bf X})({\bf S}+{\bf X})^H \right)^{p_1}\right)\right.\\
                     & \qquad \times       \mathrm{tr}\left( \left( \frac{1}{N}({\bf R}+{\bf X})({\bf S}+{\bf X})^H \right)^{p_2}\right) \cdots \\
                     & \qquad \left.\times \mathrm{tr}\left( \left( \frac{1}{N}({\bf R}+{\bf X})({\bf S}+{\bf X})^H \right)^{p_k}\right)\right].
\end{align*}
We have that
\begin{align}\label{genformulasumhere}
  M_{p_1,...,p_k} &= \sum_{{\pi\in \text{SP}_p}\atop{\pi=\pi(\rho_1,\rho_2,q)}} \frac{1}{n^kN^{|\rho_1|}} \nonumber \\
                  &\qquad\qquad \times N^{k(\rho(\pi))-kd(\rho(\pi))}  \nonumber \\
                  &\qquad\qquad \times n^{l(\rho(\pi))-ld(\rho(\pi))} n^{|\sigma|} \nonumber \\
                  &\qquad\qquad \times R_{l_1,\ldots,l_r},
\end{align}
where $l_1,\ldots,l_r$ are the cardinalities of the blocks of $\sigma$, divided by $2$.
\end{theorem}

The proof is omitted, since it follows in the same way Theorem~\ref{teo1} was generalized to Theorem~\ref{recursive} above.
Deconvolution is also possible here. It is in fact simpler than for Theorem~\ref{recursive}, in that 
there is no need to form the inverse of a matrix~\cite{ryandebbah:optstacking}. 
This is explained further in the implementation presented in Section~\ref{software}.

\subsection{Selfadjoint Gaussian matrices}
The analogues of Theorem~\ref{teo1} and Theorem~\ref{teo2} when the Gaussian matrices instead are selfadjoint look as follows.
Since Theorem~\ref{teo1} and Theorem~\ref{teo2} had straightforward generalizations to the case where all matrices are random,
we will here assume from the start that all matrices are random:
\begin{theorem} \label{selfadjprod}
Assume that the $n\times n$ random matrix ${\bf R}$ is independent from the
$n\times n$ standard selfadjoint Gaussian matrix ${\bf X}$, and define
\begin{align*}
  R_{p_1,\ldots,p_k} &=
    \E\left[ \mathrm{tr}\left( {\bf R}^{p_1} \right) \mathrm{tr}\left( {\bf R}^{p_2} \right) \cdots \mathrm{tr}\left( {\bf R}^{p_k} \right) \right] \\
  M_{p_1,\ldots,p_k} &= \E\left[ \mathrm{tr}\left(\left( {\bf R}
  \frac{1}{\sqrt{n}}{\bf X} \right)^{p_1}\right) \right.\\
                  & \qquad \times \mathrm{tr}\left(\left( {\bf R} \frac{1}{\sqrt{n}}{\bf X} \right)^{p_2}\right) \cdots \\
                  & \qquad \left.\times \mathrm{tr}\left(\left( {\bf R}
                  \frac{1}{\sqrt{n}}{\bf X} \right)^{p_k}\right)\right].
\end{align*}
Set $p=p_1+\cdots+p_k$, and let $l_1,\ldots,l_r$ be the cardinalities of the blocks of $\rho_{sa}$.
We have that
\begin{equation} \label{selfadjprodformula}
  M_{p_1,\ldots,p_k} = \sum_{{ {\pi=\pi(\rho_1,\rho_2,q)\in
  \text{SP}_p}\atop{|\rho_1|=|\rho_2|=p/2}}\atop{\rho_1,\rho_2\ \text{disjoint}}} 2^{-p/2} n^{r-p/2-k} R_{l_1,\ldots,l_r}.
\end{equation}
\end{theorem}

\begin{proof}
The proof follows in the same way as the proofs in Appendix~\ref{proofteo1} and~\ref{proofteo2}.
We therefore only give the following quick description on how the terms in (\ref{selfadjprodformula}) can be identified:
\begin{itemize}
  \item $2^{-p/2}$ comes from the $p$ normalizing factors $\frac{1}{\sqrt{2}}$
    in $\frac{1}{\sqrt{2}}({\bf Y}+{\bf Y}^H)$,
  \item $n^r$ comes from replacing the non-normalized traces with the
    normalized traces to obtain $R_{l_1,\ldots,l_r}$,
  \item $n^{-p/2}$ comes from the $p$
    normalizing factors $\frac{1}{\sqrt{n}}$ in ${\bf R} \frac{1}{\sqrt{n}}{\bf X}$,
  \item $n^{-k}$ comes from the $k$ traces taken in $M_{p_1,\ldots,p_k}$.
\end{itemize}
\end{proof}

Similarly, the result for sums involving selfadjoint matrices takes the following form:
\begin{theorem} \label{selfadjsum}
Let ${\bf X}$ be an $n\times n$ standard selfadjoint Gaussian matrix, and let ${\bf R}$ be $n\times n$ and independent from ${\bf X}$. Set
\begin{align*}
  R_{p_1,\ldots,p_k} &= \E\left[ \mathrm{tr}\left(
  \left(\frac{1}{\sqrt{n}}{\bf R} \right)^{p_1} \right) \right.\\
                  & \qquad \times \mathrm{tr}\left( \left(\frac{1}{\sqrt{n}}{\bf R} \right)^{p_2} \right) \cdots \\
                  & \qquad \left.\times \mathrm{tr}\left(
                  \left(\frac{1}{\sqrt{n}}{\bf R} \right)^{p_k} \right)
                  \right] \\
  M_{p_1,\ldots,p_k} &= \E\left[ \mathrm{tr}\left(\left(
  \frac{1}{\sqrt{n}}({\bf R}+{\bf X}) \right)^{p_1}\right) \right.\\
                  & \qquad \times \mathrm{tr}\left(\left( \frac{1}{\sqrt{n}}({\bf R}+{\bf X}) \right)^{p_2}\right) \cdots \\
                  & \qquad \left.\times \mathrm{tr}\left(\left(
                  \frac{1}{\sqrt{n}}({\bf R}+{\bf X})
                  \right)^{p_k}\right)\right].
\end{align*}
Set $p=p_1+\cdots+p_k$, and let $l_1,\ldots,l_r$ be the cardinalities of the blocks of $\sigma_{sa}$ from Definition~\ref{newsigmadef}.
We have that
\begin{align}
  M_{p_1,\ldots,p_k} &= \sum_{{ {\pi=\pi(\rho_1,\rho_2,q)\in
  \text{SP}_p}\atop{|\rho_1|=|\rho_2|\leq p/2}}\atop{\rho_1,\rho_2\ \text{
  disjoint} }} 2^{-|\rho_1|} n^{-|\rho_1|+|\rho(\pi)_{sa}|} \nonumber \\
                  & \qquad \qquad \qquad \times n^{-d(\rho(\pi)_{sa})-k+|\sigma_{sa}|} \nonumber \\
                  & \qquad \qquad \qquad \times R_{l_1,\ldots,l_r}. \label{selfadjsumformula}
\end{align}
\end{theorem}
\begin{proof}
The items in (\ref{selfadjsumformula}) are identified as follows:
\begin{itemize}
  \item $2^{-|\rho_1|}$ comes from the normalizing factors
    $\frac{1}{\sqrt{2}}$ in the $2|\rho_1|$ choices of
    $\frac{1}{\sqrt{2}}({\bf Y}+{\bf Y}^H)$,
  \item $n^{-|\rho_1|}$  comes from the normalizing
    factors $\frac{1}{\sqrt{n}}$ in the $2|\rho_1|$ choices of
    $\frac{1}{\sqrt{n}}{\bf X}$,
  \item $n^{|\rho(\pi)_{sa}|-d(\rho(\pi)_{sa})}$
    comes from counting the vertices which do not come from applications of (\ref{eq2new}),
  \item $n^{-k}$ comes from the $k$ traces taken in $M_{p_1,\ldots,p_k}$,
  \item $n^{|\sigma_{sa}|}$ comes from replacing the non-normalized traces with the
    normalized traces to obtain $R_{l_1,\ldots,l_r}$.
\end{itemize}
\end{proof}

Recursive application of theorems~\ref{recursive},~\ref{recursivesum},~\ref{selfadjprod}, and~\ref{selfadjsum},
allows us to compute moments of most combinations of independent (selfadjoint or complex) Gaussian random matrices and deterministic matrices, in any order,
and allows for deconvolution in the way explained.
This type of flexibility makes the method of moments somewhat different from that of the Stieltjes transform, 
where expressions grow more complex, when the model grows more complex. 
Moreover, contrary to methods based on the Stieltjes transform, the results scale in terms of the number of moments: 
from a given number of moments, they enable us to compute the same number of output moments. 
The theorems also enable us to compute second order moments (i.e., covariances of traces) for many types of matrices, using the same type of results.
Asymptotic properties of such second order moments have previously been studied~\cite{secondorderfreeness1,secondorderfreeness2,secondorderfreeness3}.
While previous papers allow us to compute such moments and second order moments asymptotically, in many cases the exact result is needed.

\section{Software implementation} \label{software}
Theorems~\ref{recursive},~\ref{recursivesum},~\ref{selfadjprod}, and~\ref{selfadjsum} present rather complex formulas. 
However, it is also clear that they are implementable: all that is required is traversing subsets ($\rho_1,\rho_2$), permutations ($\pi,q$), 
and implement the equivalence relations $\rho(\pi),\sigma(\pi),\rho(\pi)_{sa},\sigma(\pi)_{sa}$ from $\pi$.
Code in Matlab for doing so has been implemented for this paper~\cite{supelec:finitegaussian}, as well as the equivalence relations we have defined. 
Also, the implementation stores results from traversing all partitions in matrices, and this traversal is performed only once.
Our formulas are thus implemented by multiplying the vectors of mixed moments with a precomputed matrix. 
These operations are also vectorized, so that they can be applied to many observations simultaneously 
(each vector of mixed moments is stored as a column in a matrix, and one larger matrix multiplication is performed). 
Representing the operations through matrices also addresses more complex models, since many steps of matrix multiplication are easily combined.
In~\cite{rmtdoc}, documentation of all public functions in this library can be found, 
as well as how our methods for Gaussian matrices can be combined with other types of matrices.
The software can also generate formulas directly in \LaTeX, in addition to performing the convolution or deconvolution numerically in terms of a set of input moments.
All formulas in this section have in fact been automatically generated by this implementation. 
For products, we have written down the matrices needed for convolution and deconvolution, as described previously. 
For sums, we have only generated the expressions for the first moments. 
Due to the complexity of the expressions, it is not recommended to compute these by hand.

\subsection{Automatically generated formulas for theorems~\ref{recursive} and~\ref{teo2}}
We obtain the following expression for the first three moments in Theorem~\ref{recursive},
where $R_{p_1,...,p_k}$ (we consider only one-sided correlated Wishart matrices) and $M_{p_1,\ldots,p_k}$ are as in that theorem:
\begin{eqnarray*}
  M_1                                                                   &=& \begin{array}{c} R_1 \end{array} \\
  \left(\begin{array}{c} M_2 \\ M_{1,1} \end{array} \right)             &=& 
    \left(\begin{array}{cc} 1 & c \\ \frac{1}{cN^2} & 1 \end{array} \right) \left(\begin{array}{c} R_2 \\ R_{1,1} \end{array} \right) \\
  \left(\begin{array}{c} M_3 \\ M_{2,1} \\ M_{1,1,1} \end{array}\right) &=& 
    \left(\begin{array}{ccc} 1+\frac{1}{N^2} & 3c & c^2 \\ \frac{2}{cN^2} & 1+\frac{2}{N^2} & c \\ \frac{2}{c^2N^4} & \frac{3}{cN^2} & 1 \end{array} \right) \left(\begin{array}{c} R_3 \\ R_{2,1} \\ R_{1,1,1}\end{array} \right).
\end{eqnarray*}
More generally, in order to compute the moments of products of Wishart matrices, 
we need to compute matrices as above for the different sizes of the different Wishart matrices, and multiply these.
In Section~\ref{simulations}, we will see an example where two Gaussian matrices are multiplied.
Note that the matrices from above are not invertible when $N=1$. 

Defining
\begin{align*}
  D_p &= \mathrm{tr}\left(\left(\frac{1}{N}{\bf D}{\bf D}^H\right)^p\right) \\
  M_p &= \E\left[ \mathrm{tr}\left(\left( \frac{1}{N} ({\bf D}+{\bf X})({\bf D}+{\bf X})^H \right)^p\right) \right],
\end{align*}
in accordance with Theorem~\ref{teo2}, the implementation generates the following formulas:
\begin{eqnarray*}
  M_{1} &=& D_{1} +1\\
  M_{2} &=& D_{2} +\left(2+2c\right)D_{1} +\left(1+c\right)\\
  M_{3} &=& D_{3} +\left(3+3c\right)D_{2}\\
        & & +3cD_{1}^{2}+\left(3+9c+3c^{2}+\frac{3}{N^{2}}\right)D_{1}\\
        & & +\left(1+3c+c^{2}+\frac{1}{N^{2}}\right)\\
  M_{4} &=& D_{4} +\left(4+4c\right)D_{3} +8cD_{2}D_{1} \\
        & & +\left(6+16c+6c^{2}+\frac{16}{N^{2}}\right)D_{2}\\
        & & +\left(14c+14c^{2}\right)D_{1}^{2}\\
        & & +\left(4+24c+24c^{2}+4c^{3}+\frac{20+20c}{N^{2}}\right)D_{1}\\
        & & +\left(1+6c+6c^{2}+c^{3}+\frac{5+5c}{N^{2}}\right)
\end{eqnarray*}

\subsection{Automatically generated formulas for theorems~\ref{selfadjprod} and~\ref{selfadjsum}}
We obtain the following expression for the first four moments in Theorem~\ref{selfadjprod},
where $R_{p_1,...,p_k}$ and $M_{p_1,\ldots,p_k}$ are as in that theorem:
\begin{eqnarray*}
  \left(\begin{array}{c} M_2 \\ M_{1,1} \end{array} \right) 
  &=& \left(\begin{array}{cc} 0 & 1 \\ \frac{1}{n^2} & 0 \end{array} \right) 
      \left(\begin{array}{c} R_2 \\ R_{1,1} \end{array} \right) \\
  \left(\begin{array}{c} M_4 \\ M_{2,2} \\ M_{3,1} \\ M_{2,1,1} \\ M_{1,1,1,1} \end{array} \right)
  &=& \left( \begin{array}{ccccc} \frac{1}{n^2} & 0 & 0 & 2 & 0 \\
                                  0 & \frac{1}{n^2} & 0 & 0 & 1 \\
                                  0 & 0 & \frac{3}{n^2} & 0 & 0 \\
                                  \frac{1}{n^4} & 0 & 0 & \frac{1}{n^2} & 0 \\
                                  0 & \frac{3}{n^4} & 0 & 0 & 0 \end{array} \right)
      \left(\begin{array}{c} R_4 \\ R_{2,2} \\ R_{3,1} \\ R_{2,1,1} \\ R_{1,1,1,1} \end{array} \right)
\end{eqnarray*}
(since $M_1=M_3=M_{2,1}=M_{1,1,1}=0$). 
The implementation is also able to generate the expected moments of 
the product of any number deterministic matrices, independent, selfadjoint (or complex) Gaussian matrices, in any order~\cite{rmtdoc}. 
This is achieved by constructing the matrices for the selfadjoint and complex cases as above, 
and multiplying the corresponding matrices together in the right order.

Defining
\begin{align*}
  D_p &= \mathrm{tr}\left(\left( \frac{1}{\sqrt{n}} {\bf D} \right)^p\right) \\
  M_p &= \E\left[ \mathrm{tr}\left(\left( \frac{1}{\sqrt{n}}({\bf D}+{\bf X}) \right)^p\right)\right],
\end{align*}
in accordance with Theorem~\ref{selfadjsum}, the implementation generates the following formulas:
\begin{eqnarray*}
  M_{1} &=& D_{1} \\ 
  M_{2} &=& D_{2} +1\\
  M_{3} &=& D_{3} +3D_{1}\\
  M_{4} &=& D_{4} +4D_{2}+2D_{1}^{2}+\left(2+n^{-2}\right).
\end{eqnarray*}

\section{Applications}\label{simulations}
In this section, we consider some wireless communications examples where the presented inference framework is used.

\subsection{MIMO rate estimation} \label{MIMORate}
In many MIMO (Multiple Input Multiple Output) antenna based sounding and MIMO channel modelling applications, 
one is interested in obtaining an estimator of the rate in a noisy and mobile environment.
In this setting, one has $M$ noisy observations of the channel ${\bf Y}_i={\bf D}+\sigma{\bf N}_i$,
where ${\bf D}$ is an $n\times N$ deterministic channel matrix, 
${\bf N}_i$ is an $n\times N$ standard, complex, Gaussian matrix representing the noise, 
and $\sigma$ is the noise variance.
The channel ${\bf D}$ is supposed to stay constant during $M$ symbols. The rate estimator is given by
\begin{align}
  C &= \frac{1}{n}\log_2 \det\left( {\bf I}_n+\frac{\rho}{N} {\bf
  D}{\bf D}^{H}\right) \nonumber\\
  &= \frac{1}{n}\log_2\left(\prod_{i=1}^n(1+ \rho\lambda_i)\right), \label{channelcapacitydef}
\end{align}
where $\rho=\frac{1}{\sigma^2}$ is the SNR, and $\lambda_i$ are the eigenvalues of $\frac{1}{N}{\bf D}{\bf D}^H$.
This problem falls within the framework we are proposing. The extra parameter $\sigma$ did not appear in any of the main theorems. 
In~\cite{supelec:finitegaussian}, it is explained how this is handled by the implementation using our results. 

We would like to infer on the capacity using our moment-based framework. We are not able to find an unbiased estimator for the capacity from the moments 
due to the logarithm in (\ref{channelcapacitydef}), but we will however explain how we can obtain an unbiased estimator for the expression $\prod_{i=1}^n(1+ \rho\lambda_i)$ 
used in (\ref{channelcapacitydef}). This is simplest when a limitation on the rank, $\mathrm{rank}({\bf D}{\bf D}^H)\leq k$, is known
\footnote{In~\cite{eurecom:channelcapacity}, 
the rate was also estimated, but without actually using unbiased estimators for products of traces, as formulated in Section~\ref{section:theorems}}.
On the assumption of such a limitation, we can write $\prod_{i=1}^n(1+ \rho\lambda_i)=1+\sum_{r=1}^n \rho^r\Pi_r(\lambda_1,\cdots,\lambda_n)$, where
\begin{align*}
\Pi_1(\lambda_1,\cdots,\lambda_n)&=\lambda_1 +\cdots +\lambda_n  \\
\Pi_2(\lambda_1,\cdots,\lambda_n)&=\sum_{1\leq i<j\leq n}\lambda_i \lambda_j  \\
&\vdots&  \\
\Pi_n(\lambda_1,\cdots,\lambda_n)&=\lambda_1 \cdots \lambda_n
\end{align*}
are the elementary symmetric polynomials. 
With $D_n$ the moments of $\frac{1}{N}{\bf D}{\bf D}^H$, and $D_{\rho}$ as in Definition~\ref{ddef2}, 
The Newton-Girard formulas~\cite{book:programmingmath} (slightly rewritten) say that we can find coefficients $a_{\rho}$ so that
\[
  \Pi_k(\lambda_1,\cdots,\lambda_n) = \sum_{\rho\in{\cal P}(k)} a_{\rho} D_{\rho}.
\]
In Section~\ref{section:theorems} we explained how we can obtain unbiased estimators for the $D_{\rho}$ on the right hand side from the noisy observations ${\bf Y}_i$. 
We can thus also obtain unbiased estimators for the $\Pi_k(\lambda_1,\cdots,\lambda_n)$. 
Due to the rank restriction, only $k$ of the $\lambda_i$ are nonzero, so that only $\Pi_1,\Pi_2,...,\Pi_k$ can be nonzero. 
We thus obtain an unbiased estimator for $\prod_{i=1}^n(1+ \rho\lambda_i)$, since this can be written as a linear combination of the $\Pi_i$. 
In the following, all rate estimations will follow this strategy by first computing an unbiased estimate for $\prod_{i=1}^n(1+ \rho\lambda_i)$, 
and substituting this in (\ref{channelcapacitydef}). 
As with Theorems~\ref{recursive} and~\ref{recursivesum}, such an estimator thus scales in terms of the moments:
it depends on the first $k$ moments of the observations only, once the restriction $\mathrm{rank}({\bf D}{\bf D}^H)\leq k$ is known. 

The inference methods in Section~\ref{section:theorems} are formulated for the case of one observation only.
When we have many observations, we have some freedom in how they are combined into new estimators:
\begin{enumerate}
  \item we can form the average $\frac{1}{L}\sum_{i=1}^L ({\bf D}+\sigma{\bf N}_i)$ of the observations, and use that this has the same statistical 
    properties as ${\bf D}+\frac{\sigma}{\sqrt{L}}{\bf N}$, with ${\bf N}$ again standard, complex, Gaussian,
  \item we can stack the observations into a compund observation matrix. 
    In~\cite{ryandebbah:optstacking} it was shown how such matrices can be included in the same inference framework, so that our methods also apply to them,
  \item we can take the average of the moments we obtain from applying the framework to each observation separately.
\end{enumerate}
In~\cite{ryandebbah:optstacking}, the variances of the estimators for the moments are analyzed,
and it is shown that the two first strategies above provide lower variance than the third strategy,
and that the first two strategies have comparable variances. 
We will therefore in the following apply the framework with the first strategy. 

We have tested two cases. First a $2\times 2$-matrix
\begin{equation} \label{caseD1}
  {\bf D} = \left( \begin{array}{cc} 1 & 0 \\ 0 & 0.5 \end{array} \right)
\end{equation}
was used with $\rho=5$ and different number of observations. The corresponding simulation is shown in Figure~\ref{fig:case1}.
\begin{figure}
  \begin{center}
    \epsfig{figure=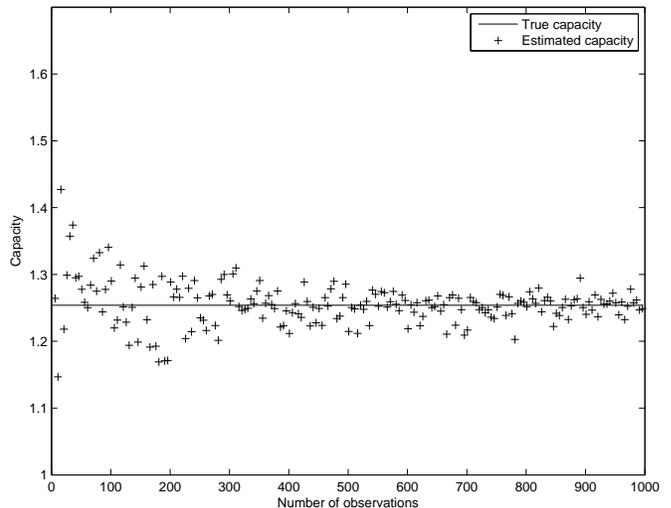,width=0.99\columnwidth}
  \end{center}
  \caption{Estimation of the channel capacity using the method of moments for the $2\times 2$-matrix (\ref{caseD1})
           for various number of observations. $\rho=5$.}\label{fig:case1}
\end{figure}
The fact that the channel matrix is diagonal is irrelevant for the rate estimation. In the second case a $4\times 4$-matrix
\begin{equation} \label{caseD2}
  {\bf D} = \left( \begin{array}{cccc} 1 & 0 & 0 & 0 \\ 0 & 0.5 & 0 & 0 \\ 0 & 0 & 2 & 0 \\ 0 & 0 & 0 & 1 \end{array} \right)
\end{equation}
was used, with $\rho=10$ and different number of observations. The corresponding simulation is shown in Figure~\ref{fig:case2}. 
In the second case, the number of variables to be estimated is higher than in the $2\times 2$-matrix case (4 eigenvalues instead of 2). 
In general one should then expect that more symbols are needed in order to obtain the same accuracy in the estimation. 
Although the figures partially confirm this, the different matrix sizes in the two cases makes the situation somewhat more involved 
(the moments converge faster for matrices of larger size).
\begin{figure}
  \begin{center}
    \epsfig{figure=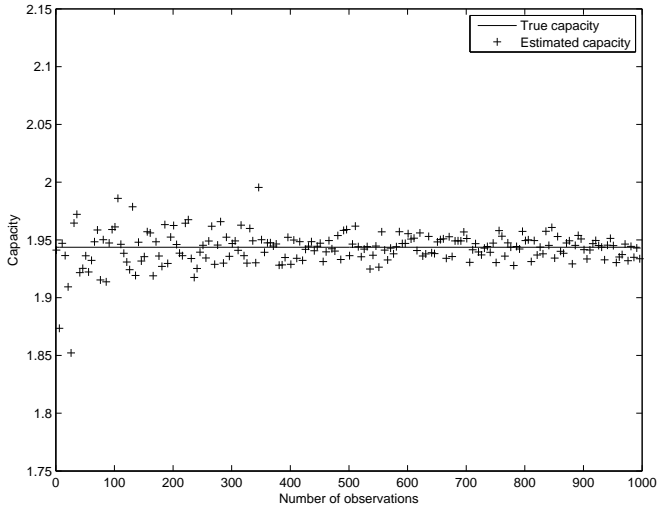,width=0.99\columnwidth}
  \end{center}
  \caption{Estimation of the channel capacity using the moment method for the $4\times 4$-matrix (\ref{caseD2})
           for various number of observations. $\rho=10$.}\label{fig:case2}
\end{figure}

\subsection{Understanding the network in a finite time}
In cognitive MIMO  Networks, one must learn and control
the \lq\lq black box\rq\rq\ (wireless channel for example) with multiple inputs and multiple outputs (Figure \ref{Feedback}) 
within a fraction of time and with finite energy. The fraction of time constraint is due to the fact that the channel (black box) changes over time. 
Of particular interest is the estimation of the rate within the window of observation.
\begin{figure}
  \begin{center}
    \epsfig{figure=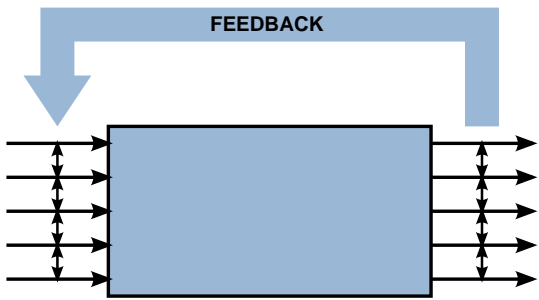,width=0.80\columnwidth}
  \end{center}
  \caption{Cognitive MIMO  Networks}\label{Feedback}
\end{figure}

Let ${\bf  y}$ be the output vector, ${\bf x}$ and ${\bf n}$ respectively the input signal and the noise vector, so that
\begin{equation} \label{DXplusN}
{\bf  y}={\bf x}+\sigma {\bf n}.
\end{equation}
In the Gaussian case, the rate is given by
\begin{align*}
C&=H({\bf y})-H({\bf y}|{\bf x})\\
&=\log_2 \det(\pi e {\bf R}_Y)-\log_2 \det(\pi e {\bf R}_N)\\
&=\log_2 \left(\frac{\det({\bf R}_Y)}{\det({\bf R}_N)}\right)
\end{align*}
where ${\bf R}_Y$ is the covariance of the output signal and ${\bf R}_N$ is the covariance of the noise. 
Therefore, one can fully  describe  the information transfer in the system knowing only the eigenvalues of ${\bf R}_Y$ and ${\bf R}_N$. 
Unfortunately, the receiver has only access to a limited number of $L$ observations of ${\bf y}$, and not the covariance of ${\bf R}_Y$. 
However, in the case where ${\bf x}$ and ${\bf n}$ are Gaussian vectors, 
${\bf y}$ can be written as ${\bf y}={\bf R}_Y^{\frac{1}{2}} {\bf u}$ where ${\bf u}$ is an i.i.d standard Gaussian vector. 
The problems falls therefore in the realm  of inference with  a correlated Wishart model 
($ \frac{1}{L}\sum_{i=1}^L {\bf y}_i {\bf y}_i^H={\bf R}_Y^{\frac{1}{2}} \frac{1}{L} \sum_{i=1}^L {\bf u}_i {\bf u}_i^H{\bf R}_Y^{\frac{1}{2}}$) .

In the simulation we have taken ${\bf n}$ as an i.i.d. standard Gaussian vector of dimension $2$, and
\begin{equation} \label{rydef}
  {\bf R}_X = \left( \begin{array}{cc} 1 & 0 \\ 0 & 0.5^2 \end{array} \right),
\end{equation}
and have used Theorem~\ref{recursivesum} to take care of the additive part, 
following up with Theorem~\ref{recursive} to take care of the Gaussian part of ${\bf x}$. 
Considering $L$ observations of (\ref{DXplusN}), we unfortunately can't use the same procedure  
as in Section~\ref{MIMORate} (i.e. averaging the observation vectors first), 
since the matrices corresponding to (\ref{recursiveformula}) are not invertible for $N=1$. 
Instead we have stacked the observations as columns in a compound matrix, and applied the framework to this 
in order to get an unbiased estimate of the moments of ${\bf R}_X$. 
In Figure~\ref{capacityDXplusN}, we have followed the same procedure 
as explained in Section~\ref{MIMORate} for estimating the capacity from these moments.
To demonstrate the convergence to the true rate, we have also increased the number of observations.
\begin{figure}
  \begin{center}
    \epsfig{figure=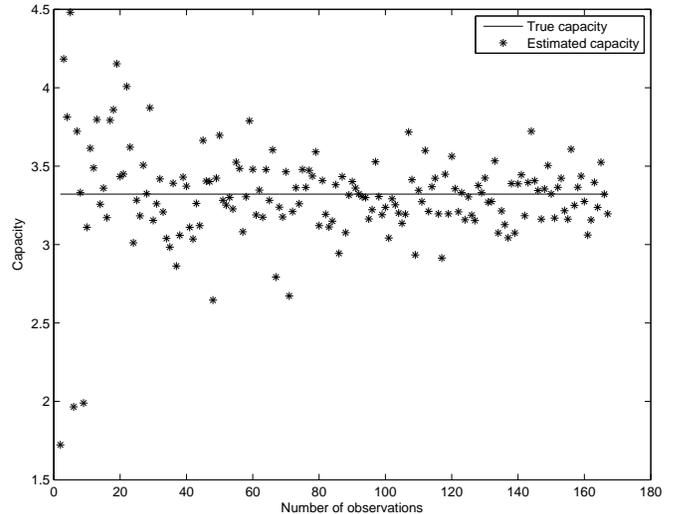,width=0.99\columnwidth}
  \end{center}
  \caption{Estimation of the capacity for the model (\ref{DXplusN}) up to $L=500$ observations, with $\sigma =0.5$.}\label{capacityDXplusN}
\end{figure}
In order to also estimate the eigenvalues of ${\bf R}_X$, 
we can first get unbiased estimates for the elementary symmetric polynomials as in section~\ref{MIMORate}, 
hence also for the characteristic equation of ${\bf R}_X$, and solve this.
Similarly to the case for the capacity, is is only the estimate for the characteristic equation which is unbiased, 
not the estimates for the eigenvalues themselves.
In Figure~\ref{simDX1plusX2_scenario1} we have shown the estimates for the eigenvalues of ${\bf R}_X$ obtained in this way.
\begin{figure}
  \begin{center}
    \epsfig{figure=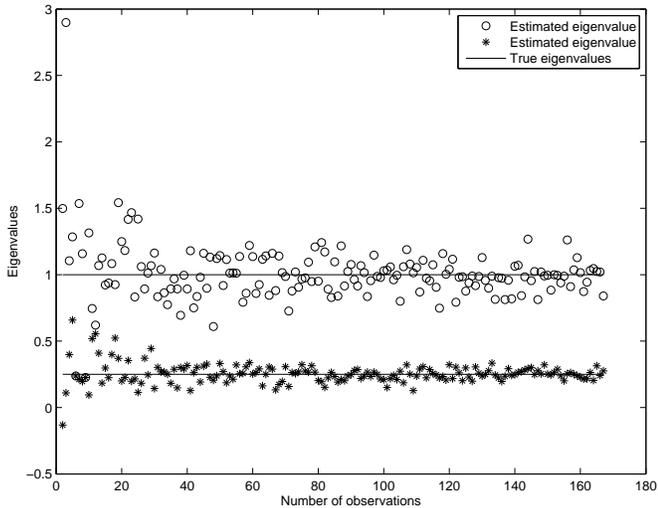,width=0.99\columnwidth}
  \end{center}
  \caption{Estimation of the eigenvalues for the $2\times 2$-matrix ${\bf R}_X$ of (\ref{rydef})
           for various number of observations.}\label{simDX1plusX2_scenario1}
\end{figure}

\subsection{Power estimation} \label{PowerEstimation}
Under the assumption of a large number of observations, our finite dimensional inference framework was not strictly needed in the two previous examples:
the observations could instead be stacked into a larger matrix, where asymptotic results are more applicable. 
When the asymptotic result can be used, inference in terms of the moments becomes simpler,
due to the almost sure convergence of the empirical eigenvalue distributions of the matrices~\cite{book:hiaipetz}.
In the asymptotic regime, Theorems~\ref{teo1},~\ref{teo2}, and~\ref{selfadjsum}
can in fact all be implemented by direct application of additive free- and multiplicative free convolution, 
and the moment-cumulant formula~\cite{book:comblect},
for which efficient implementations exists~\cite{eurecom:freedeconvinftheory}, without the need for iterating through all partitions.
Theorem~\ref{selfadjprod} can be implemented in terms of the $S$-transform~\cite{paper:vomult},
which has an implementation in terms of power series~\cite{book:tulinoverdu}, also without the need for iterating through all partitions.

This section describes a third model, where it is unclear how to apply such a stacking strategy, making the finite dimensional results more useful.
In many multi-user MIMO applications, one needs to determine the power with which the users send information. 
We consider the system given by
\begin{equation}\label{modelpower}
{\bf y}_i={\bf W}{\bf P}^{\frac{1}{2}}{\bf s}_i+\sigma{\bf n}_i
\end{equation}
where ${\bf W}$, ${\bf P}$, ${\bf s}_i$, and ${\bf n}_i$ are respectively the $N\times K$ channel gain matrix, 
the $K\times K$ diagonal power matrix due to the different distances from
which the users emit, the $K\times 1$ matrix of signals and the $N\times 1$ matrix representing the noise with variance $\sigma$. 
In particular, ${\bf W}, {\bf s}_i, {\bf n}_i$ are independent standard, complex, Gaussian matrices and vectors. 
We suppose that we have $M$ observations of the received signal ${\bf y}_i$, during which the channel gain matrix stays constant.
Considering the  $2\times 2$-matrix
\begin{equation} \label{power}
  {\bf P}^{\frac{1}{2}} = \left( \begin{array}{cc} 1 & 0 \\ 0 & 0.5 \end{array} \right),
\end{equation}
applying Theorem~\ref{recursivesum} first, and then Theorem~\ref{recursive} twice (each application takes care of one Gaussian matrix),
we can estimate the moments of the matrix ${\bf P}$ from the moments of the matrix ${\bf Y}{\bf Y}^H$, 
where ${\bf Y}=[{\bf y}_1,\ldots,{\bf y}_M]$ is the compound observation matrix. 
We assume that we have an increasing number of observations ($L$) of the matrix ${\bf Y}$, and take an average of the estimated moments 
(we average across several block fading channels). 
From the estimated moments of ${\bf P}$ we can then estimate its eigenvalues as in Section~\ref{MIMORate}.
When $L$ increases, we get a prediction of the eigenvalues which is closer to the true eigenvalues of ${\bf P}$.
Figure~\ref{simX1DX2plusX3_ObservationIncreased} illustrates the estimation of eigenvalues up to $L=1200$ observations.
\begin{figure}
  \begin{center}
    \epsfig{figure=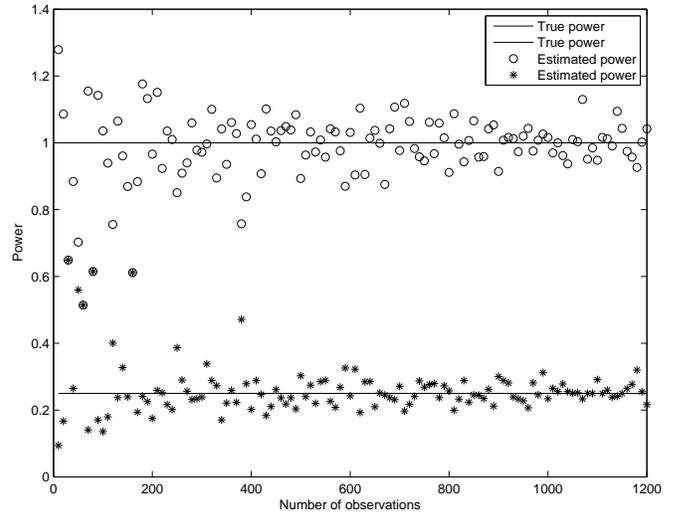,width=0.99\columnwidth}
  \end{center}
  \caption{Estimation of the powers  for the model (\ref{modelpower}),
  where the number $L$ of observations is increased, the sizes of the
  matrices are $K=N=M=2$ and $\sigma=0.1$. The actual powers are 0.25 and 1.
           }\label{simX1DX2plusX3_ObservationIncreased}
\end{figure}

It is possible to compute the variance of the moment estimators for the model (\ref{modelpower}).
We do not write down expressions for these, but remark that the framework is capable of performing this tedious task. 
These expressions turn out to involve combinations of $K$, $M$, and $N$ in the denominators,
so that in order for the variance to be low, large values for $K,M,N$ are required. 
In Figures~\ref{simX1DX2plusX3_SizeIncreased_L15}~ and~\ref{simX1DX2plusX3_SizeIncreased_L150}, 
we note that the variance decreases much faster when
we increase $K,M,N$ jointly, than when we increase the number of observations.

\begin{figure}
  \begin{center}
    \epsfig{figure=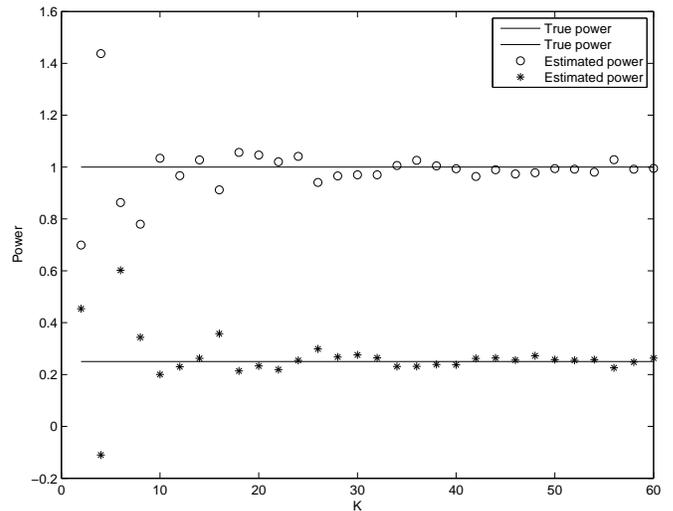,width=0.99\columnwidth}
  \end{center}
  \caption{Estimation of the powers  for the model (\ref{modelpower}), where the size $K=N=M$ of the matrices is increased, the number of observations is fixed $L=15$ and $\sigma=0.1$. The actual powers are 0.25 and 1.}\label{simX1DX2plusX3_SizeIncreased_L15}
\end{figure}

\begin{figure}
  \begin{center}
    \epsfig{figure=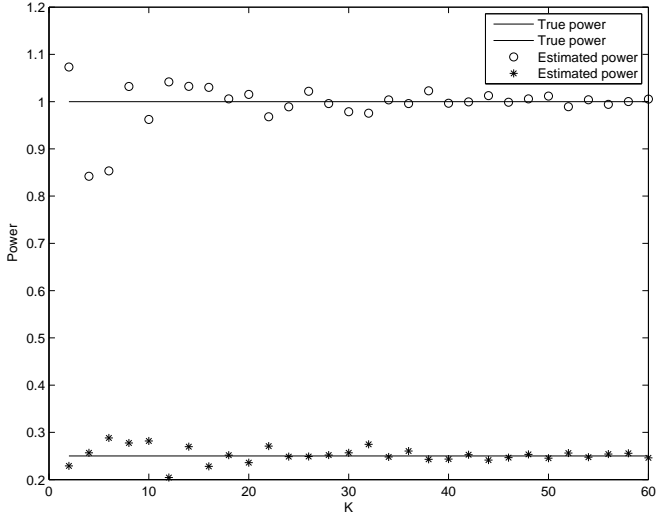,width=0.99\columnwidth}
  \end{center}
  \caption{Estimation of the powers  for the model (\ref{modelpower}), where the size $K=N=M$ of the matrices is increased, the number of observations is fixed $L=50$ and $\sigma=0.1$. The actual powers are 0.25 and 1.}\label{simX1DX2plusX3_SizeIncreased_L150}
\end{figure}

\section{Conclusion and further work}
In this paper, we have introduced a framework which enables us to compute the moments of many types of combinations of 
independent Gaussian- and Wishart random matrices, without any assumptions on the matrix dimensions. 
We also explained an accompanying software implementation, and also some useful applications where the framework has been used for simulations.

Future work will focus on applying and extending the framework to other types of matrix models. 
It may also be possible to extend the framework to obtain not only the moments we consider, but also the {\em negative moments}~\cite{paper:burda}.

While the formulas presented here have been generated by traversing sets of partitions, there may exist expressions for the same formulas which are more 
efficient to compute, as has been found at least in one case~\cite{paper:tucci1}. 
Future work will also attempt to find such simpler expressions. This is a must if the method of moments needs to compute 
moments of order much higher than used here. 

Since the method of moments only encodes information about the lower order moments, 
it lacks much information which is encoded naturally into Stieltjes transform, so that  spectrum estimation based on the Stieltjes transform 
has much better performance when few moments are considered. 
Once one can find simpler expressions for higher order moments, an interesting project would be to find how many moments are typically needed 
in order for the method of moments to perform close to the Stieltjes transform method.
methods. 

\appendices

\section{The proof of Theorem~\ref{teo1}} \label{proofteo1}
In order to prove Theorem~\ref{teo1}, we will expand the moments  
\begin{equation}\label{moments}
\E\left[\mathrm{tr}\left(\bf{D}\bf{X}{\bf E}\bf{X}^H\right)^p\right]
\end{equation}
following in the footsteps of~\cite{paper:haagerupthorbjornsen1}, 
and in the process generalize results therein, since no deterministic part was involved in that paper. 
We will thus in the following rewrite some of the important parts in the proofs in~\cite{paper:haagerupthorbjornsen1}, 
since these are needed in our generalizations. 
First, we will need the following proposition.  

\begin{proposition}\label{prop1}
Let ${\bf X}$ be $n\times N$ standard, complex, Gaussian, and ${\bf D}$ a deterministic $n\times n$ matrix.
Let $p$ be a positive integer, then
\begin{eqnarray*}\label{thesis}
\lefteqn{ \E\left[\mathrm{tr}\left({\bf D}{\bf X}{\bf E}{\bf X}^H\right)^p\right] }\\
&=& \sum_{\pi \in S_p}\E\left[\mathrm{tr}\left({\bf D}{\bf X}_{1}{\bf E}{\bf X}_{\pi(1)}^H\cdots{\bf D}{\bf X}_p{\bf E}{\bf X}_{\pi(p)}^H\right)\right]
\end{eqnarray*}
where ${\bf X}_1,\ldots,{\bf X}_p$  are independent $n\times N$
standard, complex, Gaussian matrices.
\end{proposition}
\begin{IEEEproof}
Let $\left({\bf X}_i\right)_{i\in \mathbb{N}}$ be a sequence of
independent $n\times N$ standard, complex, Gaussian matrices with entries $x(u,v,i)$, $1\leq u\leq n$, $1\leq v\leq N$. 
For any $s\in \mathbb{N}$, the matrix $s^{-1/2}\left({\bf
X}_1+\cdots+{\bf X}_s\right)$ is again $n\times N$ standard, complex, 
Gaussian. Hence, we can write
\begin{align*}
  \lefteqn{\E\left[\mathrm{tr}\left({\bf D}{\bf X}{\bf E}{\bf X}^H\right)^p\right]}
  \\
  &=\E\left\{\mathrm{tr}\left[{\bf D}\left(s^{-1/2}({\bf X}_1+\cdots+{\bf X}_s)\right)\right. \right.{\bf E}  \\
  & \qquad \left. \left.
  \left(s^{-1/2}({\bf X}_1+\cdots+{\bf X}_s)\right)^H\right]^p\right\} \\
  &=s^{-p}\sum_{1\leq i_1,j_1,\ldots ,i_p,j_p\leq s}\E\left[\mathrm{tr}\left({\bf D}{\bf X}_{i_1}{\bf E}{\bf X}_{j_1}^H\cdots {\bf D}{\bf X}_{i_p}{\bf E}{\bf X}_{j_p}^H\right)\right].
\end{align*}
Denoting by $d(i,j)$ the elements of ${\bf D}$, $e(i,j)$ the elements of ${\bf E}$, we have that
\begin{multline}\label{dE}
\E\left[\mathrm{tr}\left({\bf D}{\bf X}_{i_1}{\bf E}{\bf X}_{j_1}^H\cdots {\bf D}{\bf X}_{i_p}{\bf E}{\bf X}_{j_p}^H\right)\right]=n^{-1}\times\\
\sum_{ {1\leq u_1,u_2,\ldots,u_p\leq n \atop 1\leq v_1,v_2,\ldots,v_p\leq n} \atop {1\leq w_1,w_2,\ldots,w_p\leq N \atop 1\leq y_1,y_2,\ldots,y_p\leq N} }
 d(u_p,v_1)\cdots d(u_{p-1},v_p)\times \\
 e(w_1,y_1)\cdots e(w_p,y_p)\times \\
 E[x(v_1,w_1,i_1)\times \overline{x(u_2,y_1,j_1)} \times \cdots \\
   x(v_p,w_p,i_p)\overline{x(u_1,y_p,j_p)}],
\end{multline}
and we need only sum over conjugate pairings of the Gaussian variables, i.e. for a $\pi\in S_p$ we have
\begin{eqnarray}
  j_h          &=& i_{\pi(h)} \nonumber \\
  u_h          &=& v_{\pi(h)} \nonumber \\
  y_h          &=& w_{\pi(h)}\label{equalvariables} 
\end{eqnarray}
for all $h$.
Hence, we only have to sum over those $2$-tuples $(i_1, j_1,\ldots, i_p, j_p)$ that are in 
\begin{multline*}
M(\pi, s)=\left\{(i_1,j_1,\ldots,i_p,j_p)\in \left\{1,2,\ldots,s\right\}^{2p}|\right.  \\
 \left. j_1=i_{\pi(1)},\ldots,j_p=i_{\pi(p)} \right\}.
\end{multline*}
for some $\pi \in S_p$, i.e. 
\begin{multline*}
\E\left[\mathrm{tr}\left({\bf D}{\bf X}{\bf E}{\bf X}^H\right)^p\right]=
s^{-p}\times\\
\sum_{(i_1,j_1,\ldots ,i_p,j_p)\in \bigcup_{\pi \in S_p}M(\pi, s)}
\E\left[\mathrm{tr}\left({\bf D}{\bf X}_{i_1}{\bf E}{\bf X}_{j_1}^H\cdots {\bf D}{\bf X}_{i_p}{\bf E}{\bf X}_{j_p}^H\right)\right].
\end{multline*}
We observe that the sets $M(\pi, s)$ are not disjoint, but if we put
\begin{multline*}
\mathcal{D}(s)=\{(i_1, j_1,\ldots, i_p, j_p)\in \left\{1,2,\ldots,s\right\}^{2p}| \\
  i_1, i_2,\ldots i_p\ \mathrm{are\ distinct}\}
\end{multline*}
the sets $M(\pi, s)\cap \mathcal{D}(s)$, $\pi \in S_p$, are disjoint. Thus, we can write
\begin{multline}\label{SU}
\E\left[\mathrm{tr}\left({\bf D}{\bf X}{\bf E}{\bf X}^H\right)^p\right]= \\
s^{-p}
\hspace{-0.3cm}
\sum_{{\pi\in S_p} \atop {(i_1,j_1,\ldots ,i_p,j_p)\in M(\pi, s)\cap \mathcal{D}(s)}}
\hspace{-1cm}
\E\left[\mathrm{tr}\left({\bf D}{\bf X}_{i_1}{\bf E}{\bf X}_{j_1}^H\cdots {\bf D}{\bf X}_{i_p}{\bf E}{\bf X}_{j_p}^H\right)\right]+\\
s^{-p}
\hspace{-0.5cm}
\sum_{{\pi\in S_p} \atop {(i_1,j_1,\ldots ,i_p,j_p)\in M(\pi, s)\setminus\mathcal{D}(s)}}
\hspace{-1cm}
\E\left[\mathrm{tr}\left({\bf D}{\bf X}_{i_1}{\bf E}{\bf X}_{j_1}^H\cdots {\bf D}{\bf X}_{i_p}{\bf E}{\bf X}_{j_p}^H\right)\right].
\end{multline}
All $(i_1, j_1, \ldots, i_p, j_p)\in M(\pi, s)\cap \mathcal{D}(s)$ give the same contribution in the above sum, 
so that we can write the first term of (\ref{SU}) as
\begin{multline*}
s^{-p}\sum_{\pi\in S_p}\mathrm{card}(M(\pi, s)\cap \mathcal{D}(s))\times  \\
\E\left[\mathrm{tr}\left({\bf D}{\bf X}_{1}{\bf E}{\bf X}_{\pi(1)}^H\cdots{\bf D}{\bf X}_p{\bf E}{\bf X}_{\pi(p)}^H\right)\right].
\end{multline*}
Since the cardinality of $M(\pi, s)\cap \mathcal{D}(s)$ is equal to $s(s-1)\cdots (s-p+1)$, we have
$$\lim_{s\rightarrow \infty} s^{-p}\mathrm{card}(M(\pi, s)\cap \mathcal{D}(s)) = 1,$$
so that the first term of (\ref{SU}) tends to
$$\sum_{\pi\in S_p}\E\left[\mathrm{tr}\left({\bf D}{\bf X}_{1}{\bf E}{\bf X}_{\pi(1)}^H\cdots{\bf D}{\bf X}_p{\bf E}{\bf X}_{\pi(p)}^H\right)\right]$$
as $s\rightarrow \infty$. Observing that
\begin{align*}
& s^{-p}\mathrm{card}(M(\pi, s)\setminus\mathcal{D}(s)) \\
&=\left[ s^{-p}\mathrm{card}(M(\pi, s))-s^{-p}\mathrm{card}
    \left( M(\pi, s)\cap \mathcal{D}(s) \right)
  \right]\\
&=\left[ 1-s^{-p}\mathrm{card}
    \left( M(\pi, s)\cap \mathcal{D}(s) \right)
  \right] \longrightarrow 0,
\end{align*}
as $s\rightarrow \infty$, and summing over $\pi\in S_p$, we see that the second term in (\ref{SU}) tends to $0$, and (\ref{thesis}) follows.
\end{IEEEproof}

Theorem~\ref{teo1} will follow from Proposition~\ref{prop1}, the following proposition, 
and insertion of the additional $N^{-p}$-factor in (\ref{momentsDXX^H}):
\begin{proposition}\label{prop2}
For any positive integers $n, N$, any $\pi \in S_p$ and any ${\bf D}$ deterministic $n\times n$ matrix, we have
\begin{eqnarray}
\lefteqn{ \E\left[\mathrm{tr}\left({\bf D}{\bf X}_{1}{\bf E}{\bf X}_{\pi(1)}^H\cdots{\bf D}{\bf X}_p{\bf X}_{\pi(p)}^H\right)\right]} \nonumber \\
&=& N^{k(\rho)}n^{l(\rho)-1} D_{\rho|\text{odd}}E_{\rho|\text{even}} \label{lambda}
\end{eqnarray}
with $\rho,k(\rho),l(\rho)$ as in Definition~\ref{thisdef}.
\end{proposition}
\begin{IEEEproof}
Inserting (\ref{equalvariables}) into (\ref{dE}) we obtain
\begin{multline*}
\E\left[\mathrm{tr}\left(\left({\bf D}{\bf X}{\bf E}{\bf X}^H\right)^p\right)\right]=n^{-1}\times\\
\sum_{\pi\in S_p} \sum_{ {1\leq v_1,v_2,\ldots,v_p\leq n} \atop {1\leq y_1,y_2,\ldots,y_p\leq N} }
 d(v_{\pi(p)},v_1)\cdots d(v_{\pi(p-1)},v_p)\times \\
 e(w_{\pi^{-1}(1)},y_1)\cdots e(w_{\pi^{-1}(p)},y_p)\\
 =n^{-1}\sum_{\pi\in S_p}
 \left(\sum_{1\leq v_1,v_2,\ldots,v_p\leq n} d(v_{\pi(p)},v_1)\cdots d(v_{\pi(p-1)},v_p)\right)\\
 \times\left(\sum_{1\leq y_1,y_2,\ldots,y_p\leq N} e(y_{\pi^{-1}(1)},y_1)\cdots e(y_{\pi^{-1}(p)},y_p)\right).
\end{multline*}
The result will follow from analyzing the terms in this expression. 

$\rho$ restricted to the even numbers is generated by the relations 
\[
  2j\sim 2\pi(j), \quad j\in\{1,\ldots,p\}.
\]
Mapping even numbers $\leq 2p$ onto $\{ 1,\ldots,p\}$, this is equivalent to $j\sim \pi(j)$, $j\in\{ 1,\ldots,p\}$, i.e., 
the blocks consisting of even numbers are in one-to-one correspondence with the cycles of $\pi$. 
From this it follows that 
\begin{equation} \label{firstsum}
  \sum_{1\leq y_1,y_2,\ldots,y_p\leq N} e(y_{\pi^{-1}(1)},y_1)\cdots e(y_{\pi^{-1}(p)},y_p) = N^{k(\rho)} {\bf E}_{\rho|\text{even}},
\end{equation}
since the matrix indices follow the cycle structure of $\pi$. 
Here $N^{k(\rho)}$ comes from the fact that the summand is a product of $k(\rho)$ non-normalized traces of $N\times N$-matrices.

$\rho$ restricted to the odd numbers is generated by the relations
\[
  2j-1 \sim 2\pi^{-1}(j)+1.
\]
Mapping odd numbers $\leq 2p$ onto $\{ 1,\ldots,p\}$, this is equivalent to $j\sim \pi^{-1}(j)+1$, $j\in\{ 1,\ldots,p\}$. From this it follows that 
\begin{equation} \label{secondsum}
  \sum_{1\leq v_1,v_2,\ldots,v_p\leq n} d(v_{\pi(p)},v_1)\cdots d(v_{\pi(p-1)},v_p) = n^{l(\rho)} {\bf D}_{\rho|\text{odd}}
\end{equation}

The result now follows by inserting (\ref{firstsum}) and (\ref{secondsum}).

\end{IEEEproof}

\section{The proof of Theorem~\ref{teo2}} \label{proofteo2}
Since only conjugate pairings of Gaussian variables contribute, we need only consider partial permutations.
The contribution from the partial permutation $\pi=\pi(\rho_1,\rho_2,q)$ can be written
\begin{multline*}
n^{-1}\times\sum_{ {1\leq v_1,v_2,\ldots,v_p\leq n} \atop {1\leq w_1,w_2,\ldots,w_p\leq N} }
\prod_{i\in\rho_2^c} d(v_i,w_i) \prod_{i\in\rho_1^c} \overline{e(v_{i+1},w_i)} \times \\
\E\left[x(v_{\rho(1)},w_{\rho(1)},\rho(1))\overline{x(v_{\rho(1)+1},w_{\rho(1)},\rho_2(q(1)))}\times \cdots \right.\\
\left. x(v_{\rho_1(|\rho_1|)},w_{\rho_1(|\rho_1|)},\rho_1(|\rho_1|))\overline{x(v_{\rho_1(|\rho_1|)},w_{\rho_1(|\rho_1|)},\rho_2(q(|\rho_1|)))}\right].
\end{multline*}
Note that if $2k-1,2k\in {\cal D}$ (i.e. the first relation (\ref{eq1}) generating $\sigma$), so that $k\in\rho_1^c\cap\rho_2^c$, 
we find $d(v_k,w_k)\overline{e(v_{k+1},w_k)}$ as a part in the matrix product above, which is a part of the matrix product ${\bf D}{\bf E}^H$.
Similarly, if $2k,2k+1\in {\cal D}$, we find a part of the matrix product ${\bf E}^H{\bf D}$.

On the other hand, if $2k-1,2l\in{\cal D}$ with $(2k-1)+1=2k\sim_{\rho} 2l$ (i.e. the second relation (\ref{eq2}) generating $\sigma$),
we find that $w_k=w_l$ as in Appendix~\ref{proofteo1}, so that we find $d(v_k,w_k)\overline{e(v_{l+1},w_k)}$ as a part in the matrix product, 
which again is a part of the matrix product ${\bf D}{\bf E}^H$. We can reason similarly when $k$ and $l$ swap roles, 
to find a part of the matrix product ${\bf E}^H{\bf D}$. 

In conclusion, the relations (\ref{eq1}) and (\ref{eq2}) reflect a cyclic product of the deterministic elements, 
the length of the product equaling the number of elements in the corresponding block of $\sigma$.
Moreover, it is clear that the $\bf D$ and ${\bf E}^H$ appear in alternating order in the corresponding matrix product. In particular, 
all blocks of $\sigma$ have even cardinality.
The matrix product constitutes a non-normalized trace.
Thus, if $\sigma_i$ is the $i$'th block in $\sigma$, $|\sigma_i|$ is even, and the matrix product of the deterministic elements is
\begin{equation} \label{contribsigma}
  \prod_i \mathrm{Tr}(({\bf D}{\bf E}^H)^{|\sigma_i|/2})
  =
  n^{|\sigma|} \prod_i \mathrm{tr}(({\bf D}{\bf E}^H)^{|\sigma_i|/2}) ).
\end{equation}
(\ref{contribsigma}), which is seen to be the last term in (\ref{genformula}), 
thus contributes in $\mathrm{tr}((({\bf D}+{\bf X})({\bf E}+{\bf X})^H)^p)$.
The other terms in (\ref{genformula}) are identified as follows:
\begin{itemize}
  \item the first $n$ in the first term $\frac{1}{nN^{|\rho_1|}}$ comes from taking the trace,
    while $N^{|\rho_1|}$ comes from the normalizing factor for the Gaussian terms (the normalizing factors for the deterministic terms were absorbed in their definition).
  \item $N^{k(\rho)-kd(\rho)}$ corresponds to the number of all the
    choices of blocks of $\rho$ with even numbers only, 
    which do not intersect ${\cal D}\cup({\cal D}+1)$,
  \item $n^{l(\rho)-ld(\rho)}$ corresponds to the number of all the choices of blocks of $\rho$ with odd numbers only, 
    which do not intersect ${\cal D}\cup({\cal D}+1)$,
\end{itemize}

\bibliography{mybib,mainbib}

\begin{thebibliography}{10}
\providecommand{\url}[1]{#1}
\csname url@samestyle\endcsname
\providecommand{\newblock}{\relax}
\providecommand{\bibinfo}[2]{#2}
\providecommand{\BIBentrySTDinterwordspacing}{\spaceskip=0pt\relax}
\providecommand{\BIBentryALTinterwordstretchfactor}{4}
\providecommand{\BIBentryALTinterwordspacing}{\spaceskip=\fontdimen2\font plus
\BIBentryALTinterwordstretchfactor\fontdimen3\font minus
  \fontdimen4\font\relax}
\providecommand{\BIBforeignlanguage}[2]{{%
\expandafter\ifx\csname l@#1\endcsname\relax
\typeout{** WARNING: IEEEtran.bst: No hyphenation pattern has been}%
\typeout{** loaded for the language `#1'. Using the pattern for}%
\typeout{** the default language instead.}%
\else
\language=\csname l@#1\endcsname
\fi
#2}}
\providecommand{\BIBdecl}{\relax}
\BIBdecl

\bibitem{paper:telatar99}
E.~Telatar, ``Capacity of multi-antenna gaussian channels,'' \emph{Eur. Trans.
  Telecomm. ETT}, vol.~10, no.~6, pp. 585--596, Nov. 1999.

\bibitem{book:bouchaud}
J.-P. Bouchaud and M.~Potters, \emph{Theory of Financial Risk and Derivative
  Pricing - From Statistical Physics to Risk Management}.\hskip 1em plus 0.5em
  minus 0.4em\relax Cambridge: Cambridge University Press, 2000.

\bibitem{paper:guhr}
T.~Guhr, A.~M\"uller-Groeling, and H.~A. Weidenm\"uller, ``Random matrix
  theories in quantum physics: Common concepts,'' \emph{Phys.Rept. 299}, pp.
  189--425, 1998.

\bibitem{vo2}
D.~V. Voiculescu, ``Addition of certain non-commuting random variables,''
  \emph{J. Funct. Anal.}, vol.~66, pp. 323--335, 1986.

\bibitem{paper:vomult}
------, ``Multiplication of certain noncommuting random variables,'' \emph{J.
  Operator Theory}, vol.~18, no.~2, pp. 223--235, 1987.

\bibitem{vo6}
D.~Voiculescu, ``Circular and semicircular systems and free product factors,''
  \emph{Operator algebras, unitary representations, enveloping algebras and
  invariant theory}, vol.~92, 1990.

\bibitem{vo7}
------, ``Limit laws for random matrices and free products,'' \emph{Inv.
  Math.}, vol. 104, pp. 201--220, 1991.

\bibitem{book:hiaipetz}
F.~Hiai and D.~Petz, \emph{The Semicircle Law, Free Random Variables and
  Entropy}.\hskip 1em plus 0.5em minus 0.4em\relax American Mathematical
  Society, 2000.

\bibitem{Florent}
F.~Benaych-Georges and M.~Debbah, ``Free deconvolution: from theory to
  practice,'' \emph{submitted to IEEE Transactions on Information Theory},
  2008.

\bibitem{paper:doziersilverstein1}
B.~Dozier and J.~W. Silverstein, ``On the empirical distribution of eigenvalues
  of large dimensional information-plus-noise type matrices,'' \emph{J.
  Multivariate Anal.}, vol.~98, no.~4, pp. 678--694, 2007.

\bibitem{ryandebbah:vandermonde1}
{\O}.~Ryan and M.~Debbah, ``Asymptotic behaviour of random {V}andermonde
  matrices with entries on the unit circle,'' \emph{IEEE Trans. on Information
  Theory}, vol.~55, no.~7, pp. 3115--3148, 2009.

\bibitem{ryandebbah:vandermonde2}
------, ``Convolution operations arising from {V}andermonde matrices,''
  \emph{Submitted to IEEE Trans. on Information Theory}, 2009.

\bibitem{eurecom:freedeconvinftheory}
------, ``Free deconvolution for signal processing applications,''
  \emph{Submitted to IEEE Trans. on Information Theory}, 2007,
  http://arxiv.org/abs/cs.IT/0701025.

\bibitem{paper:haagerupthorbjornsen1}
\BIBentryALTinterwordspacing
U.~Haagerup and S.~Thorbj{\o}rnsen, ``Random matrices and {K}-theory for exact
  {$C^{\ast}$}-algebras.'' [Online]. Available:
  \url{http://citeseer.ist.psu.edu/114210.html}
\BIBentrySTDinterwordspacing

\bibitem{haagerup98random}
\BIBentryALTinterwordspacing
------, ``Random matrices with complex {G}aussian entries,'' 1998. [Online].
  Available: \url{http://citeseer.ist.psu.edu/haagerup98random.html}
\BIBentrySTDinterwordspacing

\bibitem{paper:tucci1}
G.~H. Tucci, ``A note on averages over random matrix ensembles,''
  \emph{Submitted for publication}, 2009.

\bibitem{ryandebbah:optstacking}
{\O}.~Ryan, ``On the optimal stacking of noisy observations,'' \emph{Submitted
  to IEEE Trans. Signal Process.}, 2010.

\bibitem{book:baisilverstein}
Z.~Bai and J.~W. Silverstein, \emph{Spectral Analysis of Large Dimensional
  Random Matrices}.\hskip 1em plus 0.5em minus 0.4em\relax Science Press, 2006.

\bibitem{paper:mezard}
M.~M\'{e}zard and G.~P. ans M.~Virasoro, ``Spin glass theory and beyond,''
  \emph{Physics Today}, vol.~41, pp. 1--12, 1988.

\bibitem{paper:nishimori}
H.~Nishimori, \emph{Statistical physics of spin glasses and information
  processing: an introduction}.\hskip 1em plus 0.5em minus 0.4em\relax Oxford
  University Press, USA, 2001.

\bibitem{paper:moustakasdebbah}
A.~Moustakas and M.~Debbah, ``Second-order statistics of large isometric
  matrices and applications to {MMSE SIR},'' in \emph{Asilomar conference,
  California, USA}, 2007.

\bibitem{paper:brezin}
E.~Br\'{e}zin and A.~Zee, ``Universal relation between green functions in
  random matrix theory,'' \emph{Nucl. Phys. B}, vol. 453, no.~3, pp. 531--551,
  1995.

\bibitem{paper:argaman}
N.~Argaman and A.~Zee, ``Diagrammatic theory of random scattering matrices for
  normal-metal-superconducting mesoscopic junctions,'' \emph{Phys. Rev. B},
  vol.~54, no.~10, pp. 7406--7420, Sep. 1996.

\bibitem{paper:brouwer}
P.~W. Brouwer and C.~W.~J. Beenakker, ``Diagrammatic method of integration over
  the unitary group, with applications to quantum transport in mesoscopic
  systems,'' \emph{J. Math. Phys.}, vol.~37, no.~10, pp. 4904--4933, Oct. 1996.

\bibitem{paper:bjw}
Z.~Burda, J.~Jurkiewicz, and B.~Waclaw, ``Spectral moments of correlated
  wishart matrices,'' \emph{Phys. Rev. E}, vol.~71, no.~2, 2005.

\bibitem{benaychgeorges1}
F.~Benaych-Georges, ``Rectangular random matrices. related convolution,''
  \emph{Probability Theory and Related Fields}, vol. 144, no.~3, pp. 471--515,
  2009.

\bibitem{eurecom:multfreeconv}
{\O}.~Ryan and M.~Debbah, ``Multiplicative free convolution and
  information-plus-noise type matrices,'' 2007,
  http://arxiv.org/abs/math.PR/0702342.

\bibitem{secondorderfreeness1}
J.~A. Mingo and R.~Speicher, ``Second order freeness and fluctuations of random
  matrices: {I}. {G}aussian and {W}ishart matrices and cyclic {F}ock spaces,''
  \emph{J. Funct. Anal.}, vol. 235, no.~1, pp. 226--270, 2006.

\bibitem{secondorderfreeness2}
J.~A. Mingo, P.~\'{S}niady, and R.~Speicher, ``Second order freeness and
  fluctuations of random matrices: {II}. unitary random matrices,'' \emph{Adv.
  in Math.}, vol. 209, pp. 212--240, 2007.

\bibitem{secondorderfreeness3}
B.~Collins, J.~A. Mingo, P.~\'{S}niady, and R.~Speicher, ``Second order
  freeness and fluctuations of random matrices: {III}. higher order freeness
  and free cumulants,'' \emph{Documenta Math.}, vol.~12, pp. 1--70, 2007.

\bibitem{supelec:finitegaussian}
{\O}.~Ryan, \emph{Tools for convolution with finite Gaussian matrices}, 2009,
  http://folk.uio.no/oyvindry/finitegaussian/.

\bibitem{rmtdoc}
------, \emph{Documentation for the Random Matrix Library}, 2009,
  http://folk.uio.no/oyvindry/rmt/doc.pdf.

\bibitem{eurecom:channelcapacity}
{\O}.~Ryan and M.~Debbah, ``Channel capacity estimation using free probability
  theory,'' \emph{IEEE Trans. Signal Process.}, vol.~56, no.~11, pp.
  5654--5667, November 2008.

\bibitem{book:programmingmath}
R.~Seroul and D.~O'Shea, \emph{Programming for Mathematicians}.\hskip 1em plus
  0.5em minus 0.4em\relax Springer, 2000.

\bibitem{book:comblect}
A.~Nica and R.~Speicher, \emph{Lectures on the Combinatorics of Free
  Probability}, ser. London Mathematical Society Lecture Note Series.\hskip 1em
  plus 0.5em minus 0.4em\relax Cambridge: Cambridge University Press, 2006,
  vol. 335.

\bibitem{book:tulinoverdu}
A.~M. Tulino and S.~Verd\'{u}, \emph{Random Matrix Theory and Wireless
  Communications}.\hskip 1em plus 0.5em minus 0.4em\relax
  www.nowpublishers.com, 2004.

\bibitem{paper:burda}
Z.~Burda, J.~Jurkiewicz, and M.~A. Nowak, ``Is econophysics a solid science?''
  \emph{Acta Phys. Polon. B}, vol.~34, no.~1, pp. 87--133, Jan. 2003.

\end{thebibliography}

\end{document}